\documentclass[prl,amsmath,superscriptaddress,twocolumn, longbibliography]{revtex4-1} 
\usepackage[utf8]{inputenc}

\usepackage{amssymb}
\usepackage{amsmath}
\usepackage{amsthm}
\usepackage{amsfonts}
\usepackage{bbold}
\usepackage{bm}
\usepackage{bbm}
\usepackage{braket}
\usepackage{comment}
\usepackage{color}
\usepackage{times}
\usepackage{epsfig}
\usepackage{graphicx}
\usepackage{hyperref}
\usepackage{listings}
\usepackage{soul}
\usepackage{tikz}
\usepackage{multirow}
\usepackage{algorithm}

\theoremstyle{definition}

\usepackage{verbatim}
\newtheorem{lemma}{Lemma}
\newtheorem{theorem}{Theorem}
\newtheorem{corollary}{Corollary}
\newtheorem{definition}{Definition}
\renewcommand{\min}{\mathchoice{\operatorname*{min}}{\operatorname*{min}}{\mathrm{min}}{\mathrm{min}}}

\newcommand{\abs}[1]{|{#1}|}
\newcommand{\norm}[2][]{
  \ifthenelse{\equal{#1}{}}
    {\left\| {#2} \right\|}
    {\ifthenelse{\equal{#1}{uinv}}
      {\left\vert\kern-0.25ex\left\vert\kern-0.25ex\left\vert {#2} \right\vert\kern-0.25ex\right\vert\kern-0.25ex\right\vert}
      {\left\| {#2} \right\|_{#1}}
    }
}

\newcommand{\taverage}[2][]{
  \ifthenelse{\equal{#1}{}}
  {\overline{#2}}
  {\overline{#2}^{#1}}
}

\newcommand{\tracedistance}[3][]{
  \ifthenelse{\equal{#2}{}}
  {\ifthenelse{\equal{#3}{}}
    {\mathcal{D}_{#1}}{}
  }{
    \ifthenelse{\equal{#1}{}}
    {\mathchoice{\operatorname{\mathcal{D}}\left(#2,#3\right)}{\operatorname{\mathcal{D}}(#2,#3)}{\operatorname{\mathcal{D}}(#2,#3)}{\operatorname{\mathcal{D}}(#2,#3)}}
    {\mathchoice{\operatorname{\mathcal{D}}_{#1}\left(#2,#3\right)}{\operatorname{\mathcal{D}}_{#1}(#2,#3)}{\operatorname{\mathcal{D}}_{#1}(#2,#3)}{\operatorname{\mathcal{D}}_{#1}(#2,#3)}}
  }
}
\newcommand{\fidelity}[3][]{
  \ifthenelse{\equal{#2}{}}
  {\ifthenelse{\equal{#3}{}}
    {\mathcal{F}_{#1}}{}
  }{
    \ifthenelse{\equal{#1}{}}
    {\mathchoice{\operatorname{\mathcal{F}}\left(#2,#3\right)}{\operatorname{\mathcal{F}}(#2,#3)}{\operatorname{\mathcal{F}}(#2,#3)}{\operatorname{\mathcal{F}}(#2,#3)}}
    {\mathchoice{\operatorname{\mathcal{F}}_{#1}\left(#2,#3\right)}{\operatorname{\mathcal{F}}_{#1}(#2,#3)}{\operatorname{\mathcal{F}}_{#1}(#2,#3)}{\operatorname{\mathcal{F}}_{#1}(#2,#3)}}
  }
}

\newcommand{\Sr}[3][]{
  \ifthenelse{\equal{#1}{}}
    {\operatorname{\mathnormal{S}}(#2\|#3)}
    {\operatorname{\mathnormal{S}}_{#1}(#2\|#3)}
}
\newcommand{\dd}[1]{\mathop{\mathrm{d}#1}}
\newcommand{\del}{\mathop{}\!\partial}



\newcommand{\tr}{{\rm tr}}

\begin{document}

\title{\sffamily Measuring out quasi-local integrals of motion from entanglement}
\author{Bohan Lu}
 \thanks{These two authors contributed equally.}
%\email{bohal96@zedat.fu-berlin.de}
\affiliation{Dahlem Centre for Complex Quantum Systems, Freie Universität, 14195 Berlin, Germany}
\author{Christian Bertoni}
 \thanks{These two authors contributed equally.}
\affiliation{Dahlem Centre for Complex Quantum Systems, Freie Universität, 14195 Berlin, Germany}
%\footnote{Current address: IBM Research Europe - UK, Hursley, Winchester, SO21 2JN}
\author{Steven J.~Thomson}
\email{steven.thomson@ibm.com}
\affiliation{Dahlem Centre for Complex Quantum Systems, Freie Universität, 14195 Berlin, Germany}
\affiliation{Present address: IBM Research Europe - UK, Hursley, Winchester, SO21 2JN.}
\author{Jens Eisert}
\email{jense@zedat.fu-berlin.de}
\affiliation{Dahlem Centre for Complex Quantum Systems, Freie Universität, 14195 Berlin, Germany}
\affiliation{Helmholtz Center Berlin, 14109 Berlin, Germany}

\date{\today}

\maketitle
\section{Abstract}
 {\bf  Quasi-local integrals of motion are a key concept underpinning the modern understanding of many-body localisation, an intriguing phenomenon in which interactions and disorder come together. Despite the existence of several numerical ways to compute them - and astoundingly in the light of the observation that much of the phenomenology of many properties can be derived from them - it is not obvious how to directly measure aspects of them in real quantum simulations; in fact, the smoking gun of their experimental observation is arguably still missing. In this work, we propose a way to extract the real-space properties of such quasi-local integrals of motion based on a spatially-resolved entanglement probe able to distinguish Anderson from many-body localisation from non-equilibrium dynamics. We complement these findings with a new rigorous entanglement bound and compute the relevant quantities using tensor networks. We demonstrate that the entanglement gives rise to a well-defined length scale that can be measured in experiments. }   

\section{Introduction}
It is widely believed that generic quantum systems isolated from their environments will evolve under their own dynamics until they reach an apparent equilibrium state that locally resembles the expectations of a 
thermal equilibrium state~\cite{PolkovnikovReview,christian_review}. This expectation is seen as a stepping stone
to reconcile predictions from statistical mechanics and those of basic quantum mechanics.
One major exception to this rule is the case of low-dimensional quantum systems in the presence of random disorder. Non-interacting quantum systems in one dimension will entirely fail to thermalise due to any finite concentration of disorder~\cite{Anderson58}, and in recent decades it has been shown that interacting many-body systems appear to suffer the same fate~\cite{Fleishman+80,Basko+06}, leading to the 
phenomenon now known as \emph{many-body localisation} 
(MBL)~\cite{Huse+13,HuseNandkishoreOganesyanPRB14,Altman+15,Luitz+15,Alet+18,AbaninEtAlRMP19,1409.1252}.
From a theoretical standpoint, MBL is now fairly well understood in terms of the emergence of an extensive number of conserved quantities known as \emph{(quasi-)local integrals of motion} (LIOMs, also known as localised bits or $l$-bits) which can prevent many-body systems from reaching thermal equilibrium~\cite{HuseNandkishoreOganesyanPRB14,SerbynPapicAbaninPRL13_2}. While phenomenological models based around the concept of $l$-bits have seen great success~\cite{Ros+15,Imbrie+17}, and there are several approaches that can map microscopic models onto effective $l$-bit models~\cite{Rademaker+16,Rademaker+17,Pekker+17,Goihl18,Thomson+18,Kulshreshtha+19,Thomson+20b,Thomson+20c,Thomson+21,Thomson22,Bertoni+22}, the $l$-bits themselves remain a strictly theoretical construct, 
inaccessible to any experimental probes. This is in contrast with the case of Anderson localised systems, where the exponentially localised $l$-bits can be straightforwardly related to the real-space decay of the single-particle states, which has been experimentally observed~\cite{Billy+08}.

\begin{figure}[t!]
    \centering
    \includegraphics[width=.9\linewidth]{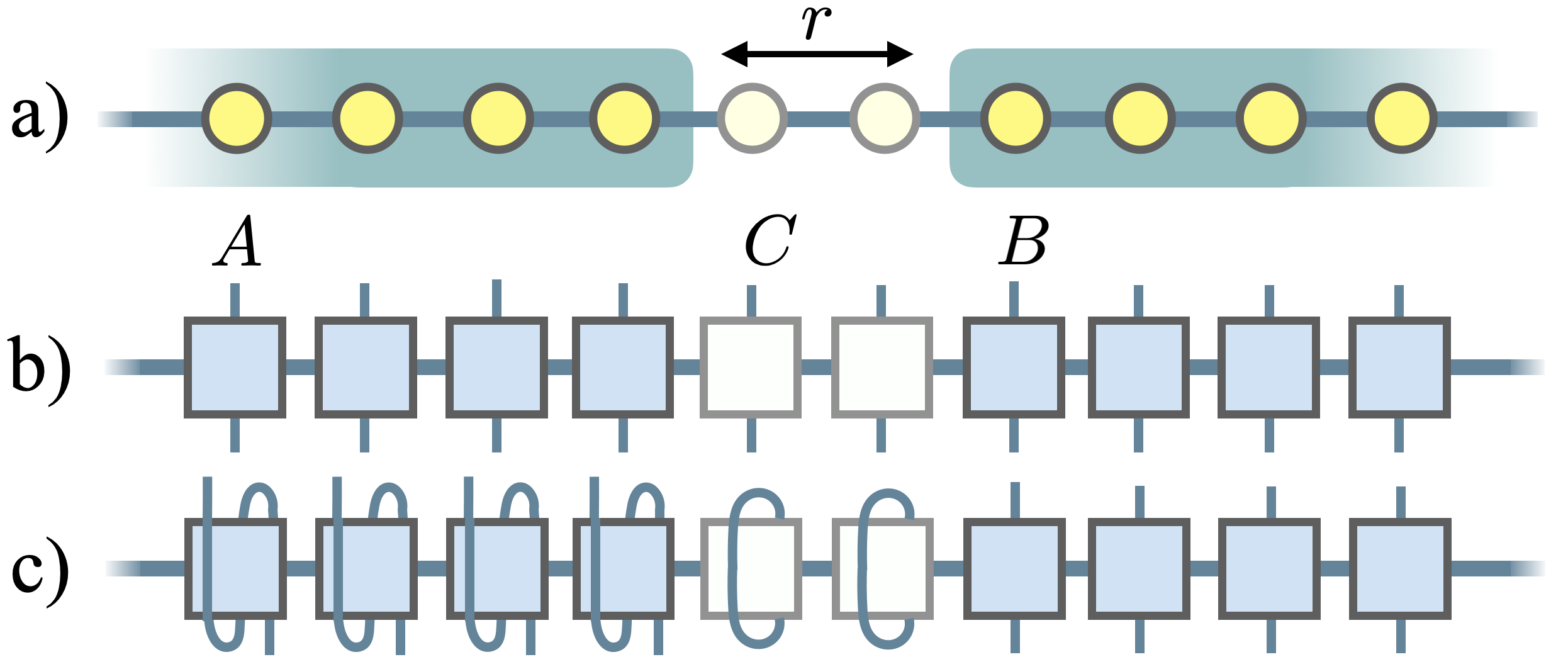}
    \caption{\textbf{Division into subsystems and computation of negativity: }a) A sketch showing how a one-dimensional spin chain is partitioned into three subsystems. We are interested in computing the entanglement between subsystems $A$ and $B$ after subsystem $C$ has been traced out, giving rise to a spatially-resolved entanglement measure.
    b) Sketch of the initial quantum state 
    in 
    \emph{matrix product operator} 
    (MPO) form, made by taking the outer product of two matrix product state vectors. c) Sketch of how the negativity is computed: the partial transpose of subsystem $A$ corresponds to `twisting' the MPO legs while tracing out subsystem $C$ corresponds to contracting the relevant MPO indices.}
    \label{fig.fig1}
\end{figure}

In this work, we propose an experimentally feasible approach to measuring the actual real-space properties of local integrals of motion in many-body quantum systems using the entanglement negativity, a sensitive entanglement monotone that allows for the recovery of spatially resolved entanglement information.  
In this way, we accommodate the above missing link.
Various quantities capturing correlations and entanglement, including the negativity, have been measured in recent experiments with ultra-cold bosons: Ref.~\cite{lukinProbingEntanglementManybody2019} has measured single-site
and half chain number and configurational entanglement
%and number entanglement 
for a system subject to a quasi-periodic potential. These can be seen as a 
witness detecting the absence of thermalization, 
but they do not provide a length
scale. The authors have also measured classical density-density correlations---akin to the proposal 
of 
Ref.~\cite{Accessible}---showing exponentially decaying correlations, but this is a two-point
classical measure, in contrast to the genuine entanglement between two half-regions
considered here.
%
%.  and has provided evidence for an exponentially decaying correlation length, 
%\newtext{by probing classical density-density correlations, akin to the suggestion of
%Ref.~\cite{Accessible}}.
%while 
Reference~\cite{Chiaro+22} has studied a disordered system and has shown that the entanglement negativity can be directly measured. 
In a first experiment, the authors of the latter work prepared the system in a product state and measured the two-qubit entanglement of formation as they vary the separation between the qubits. While this setting is close in spirit to our approach, they have chosen a two-qubit setting, which is the only setting in which one can compute this quantity, so that the diagnostic time scale that allows observation of any spatial dependence is short. In a second experiment, the preservation of entanglement has been studied, departing from the approach taken here.
Here, we demonstrate that the negativity itself gives direct access to a unique length scale that characterises the $l$-bits.

\section{Results and Discussion}
\emph{Quasi-local operators.} The question of whether many-body localisation is a well-defined stable phase in the thermodynamic limit remains unsettled, nevertheless systems showing MBL-like phenomenology for experimentally accessible times and system sizes appear to be well described by $l$-bit models, and their length scale is physically meaningful regardless of whether this description keeps holding for very long times and very large system sizes. Since the goal of the present work is to characterize the $l$-bits, we will now give precise definitions of what they are and the role they play in the dynamics.

\begin{definition}[Quasi-local operators]
	An operator $O$ on the lattice $\Lambda$ is said to be quasi-local around a region $R$ with localisation length $\xi$ if for any region $X\subset \Lambda$ containing $R$
	\begin{equation}
	    \left|\left|O- \frac{1}{2^{|X^c|}}\tr_{X^c}(O)\otimes \mathbb{I}_{X^Cc}\right|\right|^2\leq \|O\|^2\,K e^{-d(R,X^c)/\xi}
	    \label{eq.quasi-local}
	\end{equation}
	where $ K>0$ is a universal constant, $X^C$ denotes the complement of $X$, $d(X^c,R)=\min_{x\in X^c, r\in R}d(x,r)$ is the length of the shortest path from $X$ to $R$, and $\|\cdot \|$ is the normalized 2-norm, $\|O\|^2={\tr{OO^{\dagger}}}/{\tr{I}}$.
\end{definition}
Intuitively, this definition says that most of the support of $O$ is concentrated in and immediately around the region $R$, in the sense that if we truncate $O$ to an operator only supported on a sphere centered at $R$, the resulting operator differs from $O$ only by an error decaying exponentially in the radius of the sphere.
Interestingly, the relatively loose sense of decay in $2$-norm seems crucial, an insight that is often under-appreciated
\cite{goihlConstructionExactConstants2018,quasilocalProsen}. 
It is important to note that this is an abstract definition: it does not give operational advice on how to find those $l$-bits. What is more, even if they exist, they are by no means unique~\cite{Chandran+15}. There could be ``more local'' $l$-bits 
than those given that still give rise to a complete set of $l$-bits. Either way, as is common, such $l$-bits serve as our definition for many-body localisation. 

\begin{definition}[Many-body localisation]\label{def:mbl_hamiltonian} A Hamiltonian 
\begin{equation}
	H= \sum_{j=1}^n \omega_j^{(1)} h_j + \sum_{j,k=1}^n \omega^{(2)}_{j,k} h_j h_k + \dots
	\label{eq.mbl}
\end{equation}
with real weights $\{\omega_j^{(1)}\}$ and $\{\omega_{j,k}^{(2)}\}$,
is called many-body localised if it can be written as a sum of mutually commuting ($[h_j,h_k]=0$
for all $j,k$) quasi-local terms $h_j$, each centred around site $j$, and if $\omega_{i_1,\dots, i_n}\leq \omega e^{-|i_1-i_n|/\kappa}$ for some constant $\kappa>0$, where $i_1< i_2\dots < i_n$.
\end{definition}
In other words, a many-body localized Hamiltonian can be written as an effective model which is classical, in the sense that all Hamiltonian terms commute, and quasi-local, in the sense that the total coupling strength between two sites decays exponentially with distance.

{\it Premise of the approach.} 
When written in the basis that diagonalises the Hamiltonian, 
as in Eq.~(\ref{eq.mbl}), these $l$-bits are strictly local objects, but in real-space they are quasi-local with exponentially decaying tails. In order to extract properties of $l$-bits from experiments, we shall consider the evolution of an arbitrary initial state under the following Hamiltonian dynamics as
\begin{equation}
	\rho(t) = e^{-itH} \rho e^{itH},
\end{equation}
for times $t\geq 0$. To simplify the notation, we will suppress the time argument for time $t=0$. How can this time evolution be exploited to measure out real-space properties of the $l$-bits? Some intuition can be attained in the situation when the $\{h_j\}$ are strictly local. The terms that do not overlap do not contribute to the entanglement evolution at all. So in the end, it is the overlapping tails that will lead to entanglement growth. 

\emph{Model.} We will demonstrate our scheme numerically using the `standard model' of MBL, namely the XXZ spin-1/2 chain with random on-site fields, while it should be clear
that the approach taken would be applicable to
any many-body localised model. 
Its Hamiltonian is given 
by
\begin{align}
    H= J_0 \sum_{i} \left(S^{x}_i S^{x}_{i+1} + S^{y}_i S^{y}_{i+1} + \Delta S^{z}_i S^{z}_{i+1}  \right) + \sum_i h_i S^{z}_i,
\end{align}
with $h_i \in [-d,d]$. We shall set $J_0=1$ as the unit of energy throughout, with $\Delta=1.0$ unless otherwise stated, and will use open boundary conditions. This model has been thoroughly studied and shown to exhibit a phase with anomalous thermalisation properties above a disorder strength of $d \gtrsim 3.7$~\cite{Luitz+15}, although recent work has suggested that the true phase transition in the thermodynamic limit could be at much larger values of $d$ if it exists at all~\cite{Doggen+18,Suntajs+20a,Suntajs+20b,Sels+21a,Sels+21b}. 

The characteristic growth in time of the 
von Neumann entanglement entropy  
~\cite{Bardarson+12,Prosen_localisation}
(or its correlation-based analogues \cite{Accessible}) mentioned above has been shown to be a good indicator of many-body localisation, able to distinguish it from single-particle Anderson localisation via the late-time logarithmic growth.
Motivated by this, our aim in this work is to show that other entanglement measures which provide spatially-resolved information can not only distinguish many-body localisation from Anderson localisation, but can also allow direct quantitative measurement of the properties of many-body local integrals of motion. 

\emph{Diagnostic entanglement quantity.} The main quantity of interest in this work is the logarithmic negativity, a measure of the entanglement between two subsystems of the spin chain, denoted $A$ and $B$, 
separated by a distance $r$, which together with $C$ constitutes the entire system (sketched in Fig.~\ref{fig.fig1}). 
It is defined as
\cite{PhysRevA.58.883,EisertPlenioNeg,VidalNegativity,Plenio+05}
\begin{align}
    E_{N}(\rho_{A,B}(t)) &:= \log_2 (\|\rho_{A,B}^{T_A}(t)\|_{1}),
\end{align}
where $\|O\|_{1} = \tr|O|$
denotes the trace norm, $\rho_{A,B}(t)=
\tr_{\backslash \{A,B\}}[\rho(t)]$ is the time-dependent quantum state of subsystems $A$ and $B$ after tracing out all other lattice sites, and the superscript $T_A$ indicates the partial transpose with respect to subsystem $A$. This has been shown to be an entanglement monotone
meaningfully quantifying 
entanglement \cite{PhD,VidalNegativity,Plenio+05}.
In the following, we shall refer to this quantity simply as `negativity'. 
By contrast to the more commonly studied bi-partite von Neumann or Rènyi entanglement 
entropies which consider a single bi-partition between two connected subsystems, the entanglement negativity allows for a meaningful spatially resolved measure of mixed-state entanglement, as the two subsystems can be separated by an arbitrary distance 
$r :={\rm dist}(A,B)$, a feature the von Neumann entropy cannot capture as a pure state entanglement measure. 
This measure can also be used to study the entanglement between subsystems of arbitrary size. 
However, for conceptual clarity, we shall mainly consider $A$ and $B$ to cover the entire chain except for
a piece $C$ with $|C|=r+1$ separating $A$ and $B$, as shown in Fig.~\ref{fig.fig1}. That said, the
concept works as well for small regions $A$ and $B$, as they are accessible in experiments and
are discussed in the rigorous bounds. Numerical evidence is shown in Supplementary Note 5. The negativity has previously been investigated in the context of ground states of disordered spin chains~\cite{Ruggiero+16}, quenches in random spin chains~\cite{Ruggiero+22}, the many-body localisation transition~\cite{grayScaleInvariantEntanglement2019}, and quench dynamics in the presence of a defect~\cite{Gruber+20}. 
For clarity, in the following, we shall drop 
the explicit dependence of $E_N$ on the quantum state and instead use the notation $E_N(r,t)$ to represent the negativity associated to two subsystems separated by a distance $r$ and a time $t$ following a quench from an initial product state, emphasizing that this is indeed a spatially resolving 
entanglement measure.

A heuristic argument for why this quantity is relevant in our case can be given in the following manner. Ref.~\cite{SerbynPapicAbaninPRL13_2} 
has shown that the von Neumann entanglement entropy grows in time following a quench according to $S_{\rm ent} \propto \ln(J_0 t/\hbar)$, 
once the system enters the late-time equilibration regime. If we wish to consider the entanglement negativity between two subsystems separated by a distance $r$, a reasonable starting assumption is that the negativity will vary in time according to the same $\sim \ln(t)$ growth but will be exponentially suppressed in magnitude due to the spatial separation of the two subsystems, leading to an overall behaviour of $E_N \propto \exp(-r/\xi) \ln(J_0 t/\hbar)$.  We shall show that this 
 ansatz is a good match for the numerical results. We also wish to emphasise that this logarithmic growth is characteristic of the interacting system and is entirely absent from Anderson-localised systems, meaning that the existence of this length scale is a distinct fingerprint of a many-body localised system.
 \smallskip

\emph{Corroborating the reasoning with rigorous bounds.} We see that 
Hamiltonians that are many-body localised in the sense of Definition 2 create entanglement at a rate that 
decays exponentially in the distance $r={\rm dist}(A,B)$ between parts $A$ and $B$, reflecting 
the exponential decay of the tails in quasi-local $l$-bits. In fact, not only 
this intuition can be made entirely rigorous, but, at the cost of slightly weakening the definition of quasi-locality, we are in the position to 
state precise upper bounds for the
negativity for all times and distances.

\begin{theorem}[Rigorous entanglement bounds] 
Let $\rho$ be an initial product state. Let $H$ be a many-body localised Hamiltonian as per Definition \ref{def:mbl_hamiltonian} with localisation length $\xi<{1}/({4\log(2)})$ and $2(1/\kappa-\log(2))> 1/\xi$, consider three blocks $A,C,B$ such that $C$ divides $A$ from $B$, with $|C|=r+1$. The
growth of the  negativity of the state
$\rho(t) = e^{-itH}\rho e^{it H}$ restricted to the regions $A,B$ is bounded as 
\begin{equation}
    E_N(r,t)\leq \min\{t\,O(e^{-r/(2\xi)}), 8\xi\log_2(t)-2r\}+O(1),
    \end{equation}
for times $t\geq e^{r/(4\xi)}$, while for $t<e^{r/(4\xi)}$,
\begin{equation}
    E_N(r,t)\leq t\,O(e^{-r/(4\xi)}).
    \end{equation}
\end{theorem}
We hence find a short time behaviour signifying a linear growth in time, a
cross-over regime governed by the correlation length, and a logarithmic
growth for long times. These bounds---interesting in their own right and
complementing and refining those of Ref.~\cite{kimLocalIntegralsMotion2014}---are
perfectly compatible with the above numerical assessment.
In Supplementary Note 8, we state details of the 
proof
of the bound that makes extensive use of the precise 
form of the tails
of the $l$-bits. Based on our numerical results, we expect that our assumptions on the localisation length and the definition of quasi-locality can be relaxed without affecting the result. We also note that the observed $\xi-$dependence of the late time decay of the entanglement with the size of $C$ is not visible in this bound, though we expect that it can be refined to show this.
\begin{figure}[t]
    \centering
    \includegraphics[width=\linewidth]{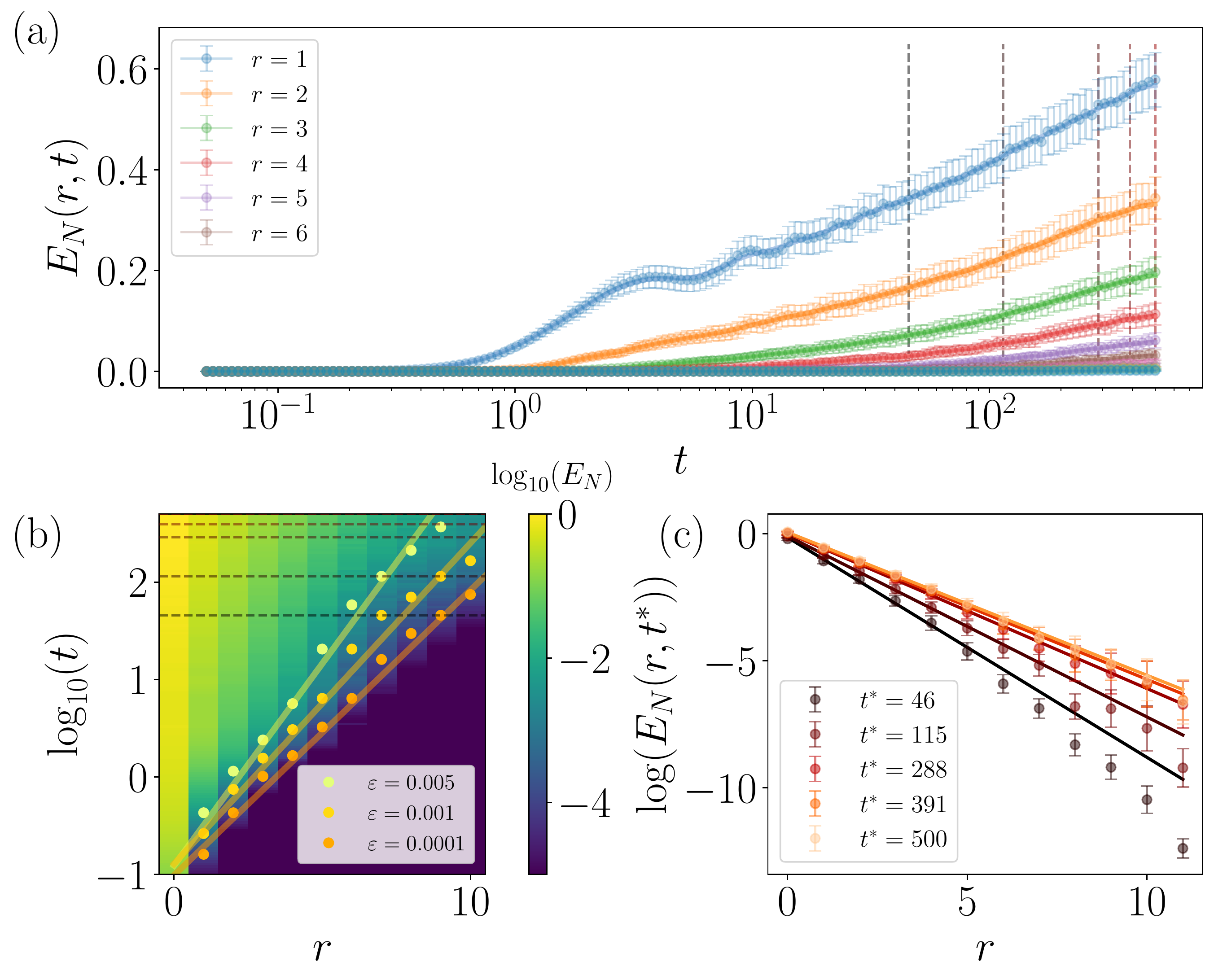}
    \caption{\textbf{Behaviour of entanglement negativity in time and space:} Results showing the growth of the negativity $E_N(r,t)$ with time for different distances $r$. Data is shown for a system size $L=24$ and a disorder strength $d=8.0$, averaged over $N_s=100$ disorder realisations. a) The dynamics of $E_N(r,t)$ following a quench from a Nèel state, showing the logarithmic growth at late times. The circular markers are the raw data points, while the solid lines are a smoothed guide to the eye. The error bars indicate the standard error in the mean. We note that these error bars show agreement on average between the various disorder realizations, but they are not fully statistically independent errors, as would be expected in an experiment where each data point would come from a different run. b) The full dynamics of $E_N(r,t)$, reflecting the logarithmic `light cone'. Each circle maps the point where the negativity grows beyond the corresponding threshold $\varepsilon$ and the lines are linear fits.  c) By extracting the behaviour of $E_N(r,t^*) \propto \exp(-r/\xi)$ at fixed times $t^*$ [dashed vertical lines in panel (a), horizontal lines in panel (b)], we can extract a well-defined length scale $\xi(t)$, which depends only weakly on time. The solid lines indicate the fits to the data points which are used to extract the $l$-bit length scale, demonstrating convergence at late times.}
    \label{fig.fig2}
\end{figure}

\emph{Numerical results.} We first discuss qualitatively the results for the growth of the entanglement negativity with time for various different distances $r$, as shown in Fig.~\ref{fig.fig2}a) for a disorder strength $d=8.0$ (deep in the localised phase), where we find that indeed the negativity grows logarithmically with time at late times. 
Results for further disorder strengths, system sizes, and subsystem sizes are available in Supplementary Notes 2,3,4 and 5. At short times, the negativity is dominated by diffusive transport on length scales shorter than the localisation length. At large distances $r$, the negativity remains close to zero until a time exponentially large in $r$, which can be used to define a `light cone' that characterises the spreading of the entanglement negativity, shown in Fig.~\ref{fig.fig2}b). The three lines indicate when the negativity grows above a threshold $\varepsilon \in \{ 0.0001,0.001,0.005 \}$, mapping out an approximately logarithmic light cone. As the negativity outside of this length cone is exponentially small, in the following analysis, we restrict ourselves to space-time coordinates $(r,t)$, which are within the light cone. The existence of this light cone means that we gain only diminishing returns by going to larger system sizes: although we are able to separate the subsystems by a larger value of $r$, the evolution time required to obtain meaningful entanglement scales exponentially in $r$, which incurs a large computational cost for large systems and quickly becomes prohibitive.

In the late-time logarithmic growth regime, where the dynamics are dominated by the quasi-local nature of the $l$-bits, we extract the value of the negativity at a given time $t^*$ following the quench from an initial Néel state and plot it versus the subsystem separation $r$. We show this in Fig.~\ref{fig.fig2}c) for several different choices of time $t^*$ [indicated by the dashed lines in Fig.~\ref{fig.fig2}a)]. The data points form a straight line (on a logarithmic scale), and at late times the gradient of the line does not strongly change with the choice of time $t^*$, appearing to saturate at a fixed value (although the y-axis offset will, of course, continue to increase in time). Further details are available in Ref. Under the assumption that the 
negativity decays exponentially with distance like $E_N(r,t^*) 
\propto \exp(-r/\xi)$, we can perform a linear fit to the data shown in Fig.~\ref{fig.fig2}c) and extract a well-defined length scale $\xi$ which characterises the spatial extent of the $l$-bits. The results are shown in Fig.~\ref{fig.fig3}, where we find that the length scale $\xi$ exhibits monotonic decay with increasing disorder strength, as expected. Note that no assumptions are involved other than the exponential decay of the negativity with distance at some fixed time $t^*$: the resulting length scale is an emergent property of the many-body system. This assumption does not hold in the delocalised phase, where the entanglement does not enter a regime of logarithmic growth.
We can further compare the length scale extracted from our procedure with the $l$-bit decay lengths computed using the established numerically exact method of Ref.~\cite{goihlConstructionExactConstants2018}, using the definition of quasi-locality from Eq.~(\ref{eq.quasi-local}). We find excellent agreement between the entanglement-based length scale and the $l$-bit localisation length obtained independently from this method, confirming that the length scale probed by the negativity is the localisation length of the $l$-bits.

\begin{figure}[t]
    \centering
    \includegraphics[width=\linewidth]{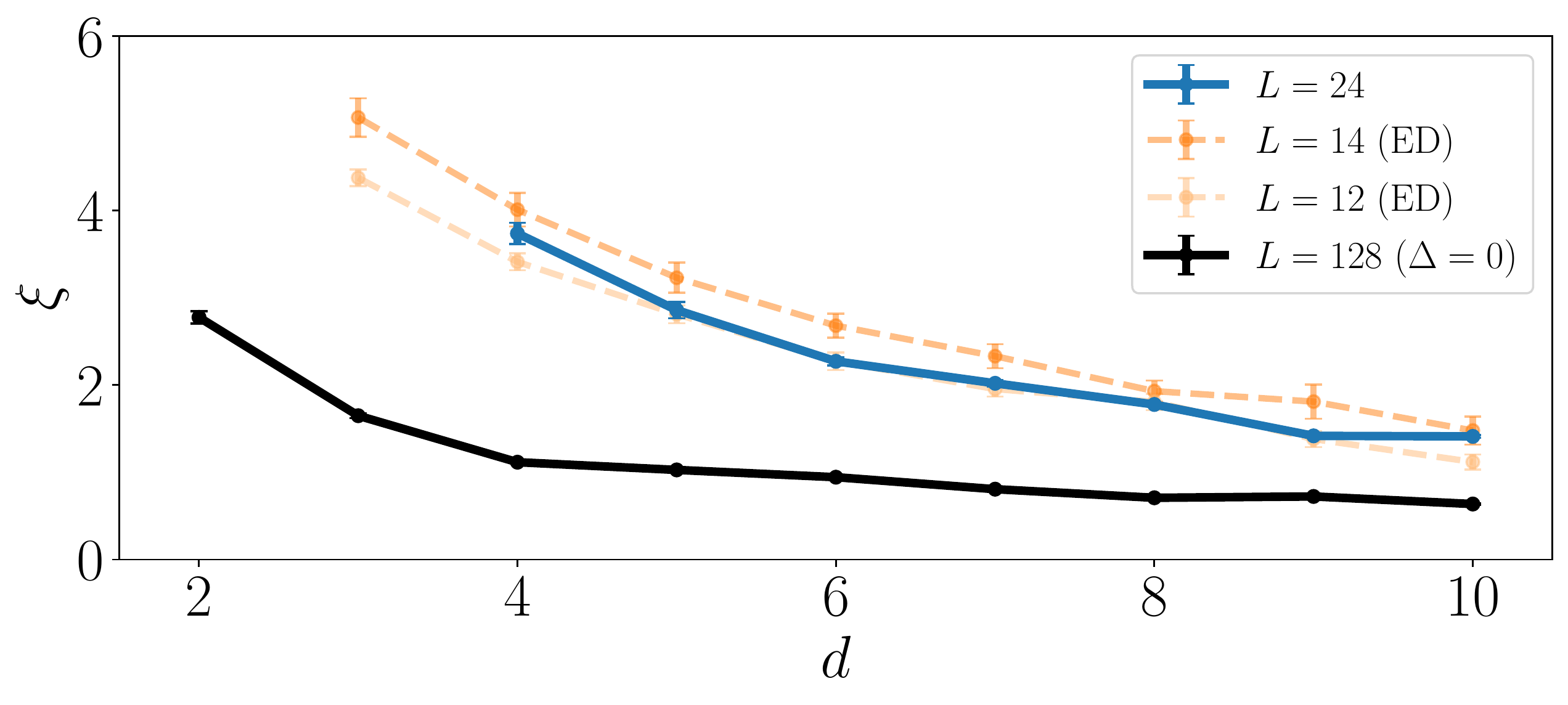}
    \caption{\textbf{Extracted l-bit length scale:} The characteristic $l$-bit length scale $\xi$ extracted from the entanglement negativity at time $t^* = 500$, shown in blue for $L=24$ with $N_s \in [50,100]$ disorder realisations and various values of the disorder strength $d$. Error bars indicate the fit error and are roughly the same size as the plot markers. Orange lines mark the localisation length obtained through exact diagonalisation following Ref.~\cite{goihlConstructionExactConstants2018}. For further details on calculating the localisation length using exact diagonalisation, see Supplementary Note 7 \cite{SM}. The black line indicates the localisation length of the corresponding Anderson-localised system, obtained by directly diagonalising the Hamiltonian in the non-interacting limit ($\Delta=0$), for a system of size $L=128$ with $N_s=10000$.}
    \label{fig.fig3}
\end{figure}

For comparison, we also indicate the corresponding 
localisation length of an Anderson localised system, here obtained by directly diagonalising the Hamiltonian with $\Delta=0$ (following a Jordan-Wigner transform into the fermionic representation). We compute the
eigenvectors of the Hamiltonian in the non-interacting setting,
which decay in real space as
$\exp(-r/\xi)$~\cite{Lee+85}, average over disorder realisations and extract the localisation length $\xi$ from a least-squares fit.
The length scale extracted from the TEBD data behaves in a qualitatively similar manner to the single-particle localisation length but is always larger, confirming that we are not measuring single-particle properties but are indeed extracting a genuinely many-body feature of the system. In the 
delocalised phase, our assumed form of the negativity is no longer valid, and as such, the method cannot extract a reliable length scale.

We also note that the entanglement negativity is not the only entanglement measure which may be used in this way: any spatially-resolved entanglement probe should behave similarly. In Supplementary Note 6, we demonstrate that the mutual information also gives consistent results.

\section{Conclusion}
In this work, we have outlined an experimentally feasible procedure for measuring local integrals of motion based on their contribution to the slow growth of the negativity at long times following a quench from an arbitrary initial state. We have demonstrated that the length scale which we obtain from this procedure, which characterises the $l$-bits, is in good agreement with that obtained using other theoretical methods in the literature. The crucial advantage is that our scheme is experimentally tractable, unlike other purely theoretical/numerical methods, which cannot be verified in real experiments. It would be extremely interesting to apply this method to other scenarios where many-body localisation is believed to exist, such as in disorder-free systems and two-dimensional models, in order to see if well-defined length scales based on the spreading of entanglement may still be extracted in these situations. This work paves the way for the application of spatially-resolved entanglement probes to phenomena in quantum simulation beyond many-body localisation, where such methods may be able to provide valuable insight into emergent length scales associated with other types of quasi-particles.

\section{Methods}
We compute the negativity using time-dependent matrix product state simulations
-- an instance of a tensor network method 
\cite{Orus-AnnPhys-2014}
-- implemented in the \texttt{Quimb} package~\cite{quimb} using the 
\emph{time-evolving block decimation} (TEBD) 
algorithm to perform the evolution~\cite{Vidal04,MPSRev}, with the system initially prepared in a Néel state.
We use system size $L=24$ with a maximum bond dimension of $\chi=192$. We perform the time evolution using a maximum time step $dt=0.05$, at each step discarding singular values smaller than $\epsilon = 10^{-10}$. We have checked that the results are well-converged. 
Detailed benchmarks are shown in Supplementary Note 1. 
Our TEBD results are compared against $l$-bit length scales obtained using exact diagonalisation, following Ref.~\cite{goihlConstructionExactConstants2018}.

The negativity can be computed straightforwardly from a 
\emph{matrix product state} (MPS) representation~\cite{grayFastComputationManyBody2018a}. The 
state vector can be turned into a \emph{matrix product operator} (MPO) (sketched in Fig.~\ref{fig.fig1}) representing the quantum state by considering vectors and dual vectors
represented as MPS. The partial transpose can be computed by `twisting' the legs of the MPO tensors, while the partial trace over the subsystem $C$ can be performed by contracting the free indices of the MPO tensors in this subsystem. At long times, the negativity should saturate at a value controlled by the size of the subsystems, and at any time $t < t^s$ (where $t^s$ is the saturation time), the negativity should satisfy the hierarchy $E_N(r_1,t) < E_N(r_2,t)$ for any two distances $r_1 > r_2$.
\section{Data availability}
The full data for this work is available at Ref.~\cite{data}.

\section{Code availability}
The full code for this work is available at Ref.~\cite{code}.

\section{Author contributions}

JE initially conceived the project. BL and JE jointly developed the idea of using entanglement bounds. SJT and BL wrote the code and performed the simulations. CB and JE proved the entanglement growth bound. All authors contributed to the writing of the final manuscript.

\section{Acknowledgements}
This project has been inspired by discussions with
P.~Roushan and B.~Chiaro of Google AI. We also thank 
A.~Kshetrimayum, S.~Sotiriadis, S.~Qasim, and J.~Gray for discussions. B.~Lu is grateful for feedback from D.~Abanin, M.~Fleischhauer, and M.~Kiefer-Emmanouilidis at the CRC 183 
summer school ``Many-body physics with Rydberg atoms". We gatefully acknowledge D. Toniolo for finding an error in an earlier draft of this work. This project has received funding from the European Union’s Horizon 2020 research and innovation programme under the Marie Skłodowska-Curie grant agreement No.~101031489 (Ergodicity Breaking in Quantum Matter), the Quantum Flagship (PASQuanS2),the Munich Quantum Valley,  the Deutsche Forschungsgemeinschaft (CRC 183 and FOR 2724), and the ERC (DebuQC). 
We also acknowledge funding from the BMBF (FermiQP and MUNIQC-ATOMS).

%%%%APPENDIX NOW IN SEPARATE FILE CALLED CommunicationPhysicsSupplementaryInfo.tex

\pagebreak
\clearpage
\onecolumngrid
\appendix

\section{\large{Supplementary information}}
%%sets figures counter to 0 and adds S in front of the number 
\setcounter{figure}{0}
\renewcommand{\figurename}{Fig.}
\renewcommand{\thefigure}{S\arabic{figure}}
\section{Supplementary note 1: TEBD accuracy benchmarks}
As TEBD does not precisely conserve the energy of the initial state, as a benchmark of the accuracy of our numerics we 
compute the relative error in the energy of the time-evolved state vector, 
$E(t) = \braket{\psi(t) |H | \psi(t)}$, computed with respect to its energy at time $t=0$. The relative error is defined as
\begin{align}
	\delta E(t) := \left| \frac{E(t) - E(t=0)}{E(t=0)}  \right|.
\end{align}
Fig.~\ref{fig.SM_energy} 
shows the relative error versus time for a variety of different disorder strengths. We find that in most cases, the relative error remains close to $\delta E \approx 0.001$, i.e., 
the energy is conserved up to an error of approximately one-tenth of a percent, confirming that our simulations are reliable.

\begin{figure}[h]
	\centering
	\includegraphics[width=\linewidth]{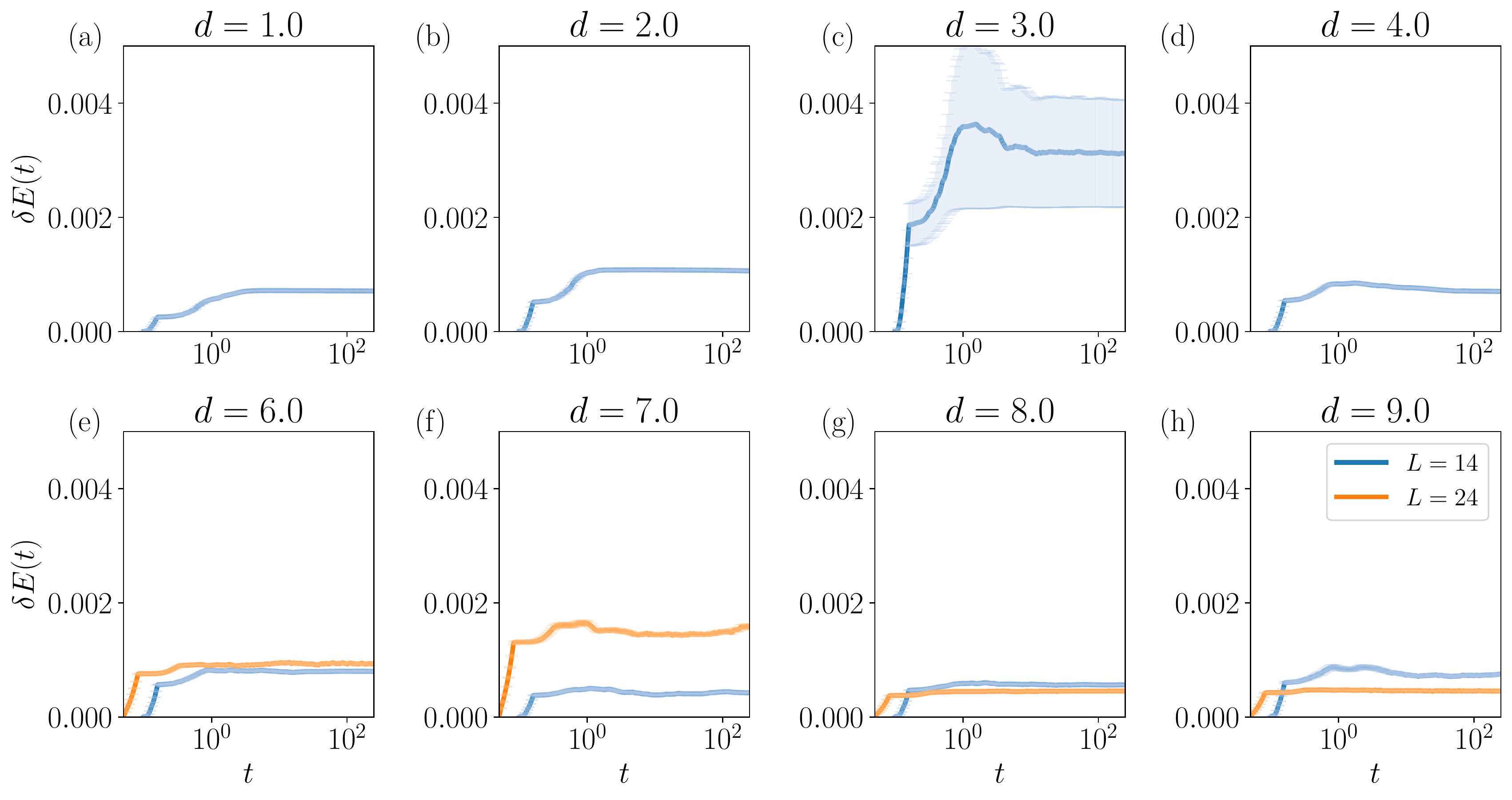}
	\caption{\textbf{Relative error of the simulation:} A comparison of the relative error in the energy of the time-evolved state for different values of the disorder strength $d$, shown for system sizes $L=14$ ($\chi=128$, averaged over $N_s=240$ disorder realisations) and - at strong disorder only - $L=24$ ($\chi=192$, averaged over $N_s=100$ disorder realisations). The relative error remains below $1\%$ for all disorder strengths. Error bars indicate the variance over disorder distributions and in most cases, are of comparable size to the plot markers.}
	\label{fig.SM_energy}
\end{figure}

\section{Supplementary note 2: Comparison of disorder strengths}

In addition to the data presented in the main text, here we show the behaviour of the entanglement negativity over a range of different disorder strengths, demonstrating the logarithmic growth at late times in the localised phase, and the qualitatively different behaviour seen in the delocalised phase. The results are shown in Fig.~\ref{fig.SM_comparison}, for a system of size $L=14$, bond dimension $\chi=128$, and averaged over $N_s=240$ disorder realisations. The circle markers represent the data points, while the solid lines are smoothed guides to the eye.
Deep in the delocalised phase ($d=1.0$), we see that the negativity saturates for all values of $r$ at relatively early times, making it difficult to pinpoint a regime where the growth of the negativity can be associated with a length scale. In contrast, the negativity in the localised phase increases much more slowly with time, and the spacing of the curves is consistent with an exponential suppression of the negativity with distance, as demonstrated in the main text.

\begin{figure}[h]
	\centering
	\includegraphics[width=\linewidth]{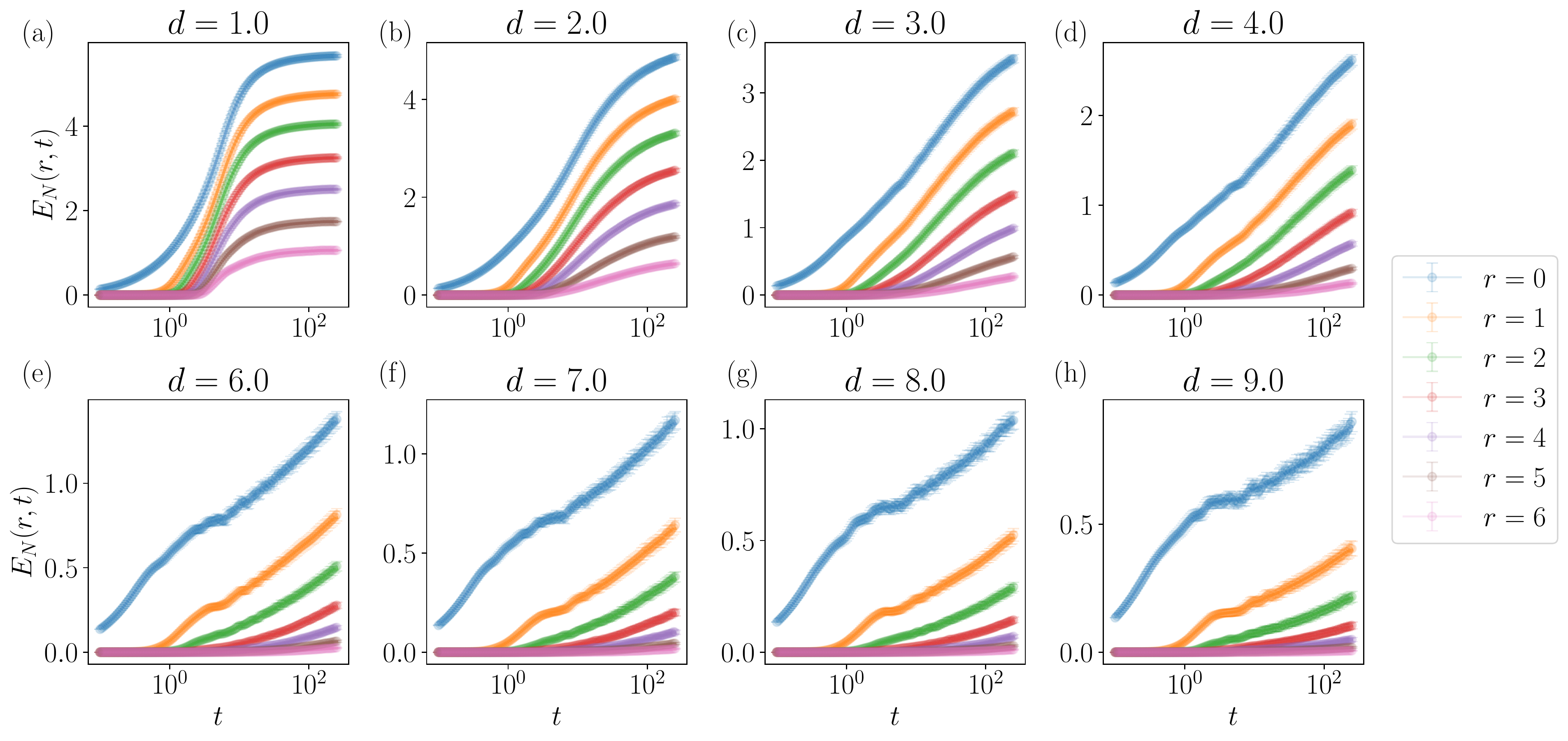}
	\caption{\textbf{Entanglement negativity at different disorder strengths: }A comparison of the dynamics of the entanglement negativity $E_N(r,t)$ for different values of the disorder strength $d$, shown for $L=14$ with bond dimension $\chi = 128$ and averaged over $N_s=240$ disorder realisations. In the delocalised phase, the negativity saturates to a value determined by the size of the subsystems $A$ and $B$, while in the localised phase the negativity displays a slow $\propto \log(t)$ growth even at late times. In this dephasing regime, we are able to use the data shown here to extract a length scale that characterises the localised phase, as detailed in the main text.}
	\label{fig.SM_comparison}
\end{figure}

\section{Supplementary note 3: Bond dimension}

In Fig.~\ref{fig.SM_chi}, we show the negativity dynamics for system size $L=24$ and varying bond dimension $\chi$, demonstrating that for $\chi=192$ (the choice used in the main text) the results are well-converged. The crucial factor for our work is the rate of growth of the negativity, which appears largely unaffected by the choice of bond dimension, although deviations can be seen for the smallest value shown in Fig.~\ref{fig.SM_chi}.

\begin{figure}[ht]
	\centering
	\includegraphics[width=0.8\linewidth]{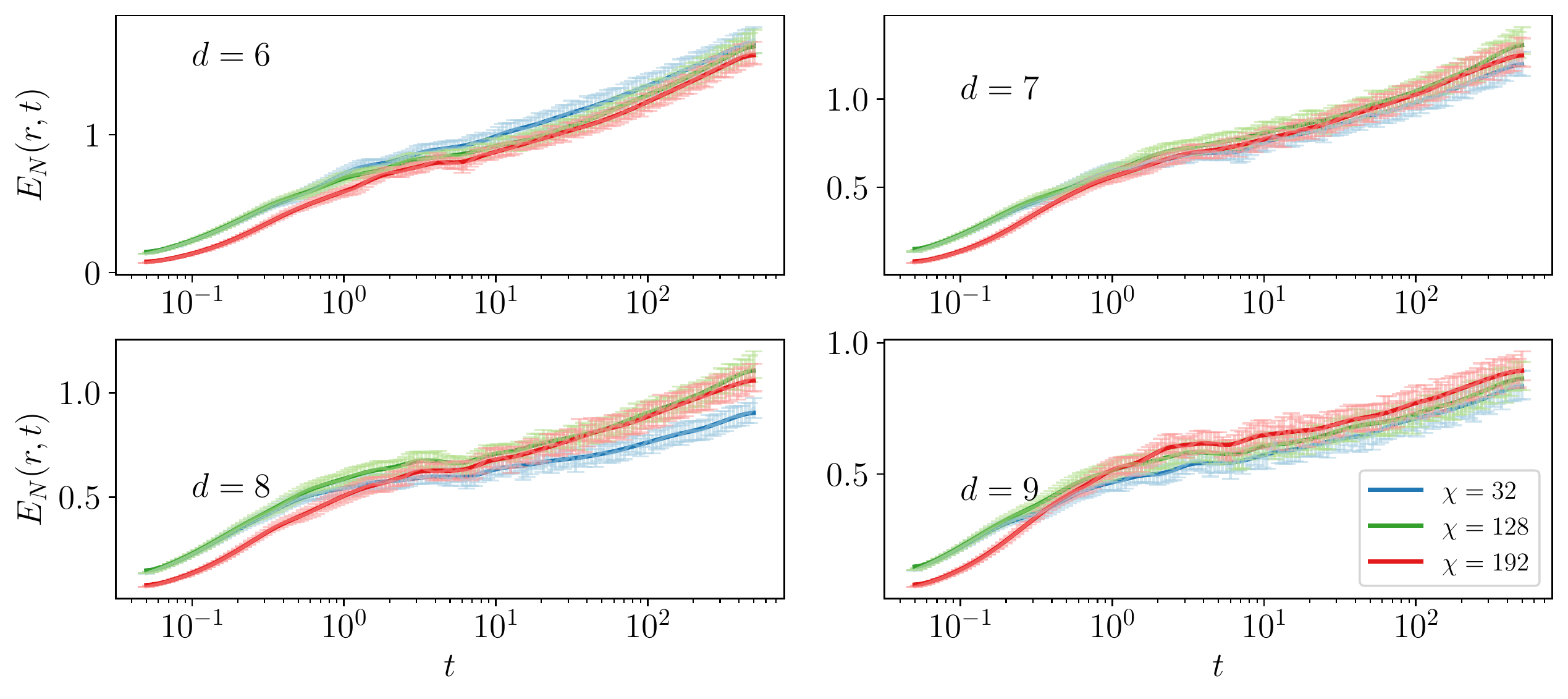}
	\caption{\textbf{Entanglement negativity for different bond dimensions and disorder strengths: }The dynamics of the negativity $E_N(r,t)$ with $r=0$, for different bond dimensions and disorder strengths. Data shown is for $L=24$, averaged over $N_s = 100$ disorder realisations. Error bars show the standard error.}
	\label{fig.SM_chi}
\end{figure}

\section{Supplementary note 4:Comparison of different measurement times}

In the main text, all results for the length scale $\xi \equiv \xi(t^*)$ are taken with $t^*=500$, i.e., 
the maximum evolution time of our simulations, however, it is clear from the negativity dynamics that there should be a weak dependence of $\xi$ on the measurement time $t^*$.
Here we demonstrate this effect. Figs.~\ref{fig.SM_fits} and \ref{fig.SM_fits2} show how the linear fit used to extract the decay of $E_N(r,t)$ against $r$ depends on the choice of time $t^*$ for a variety of different disorder strengths, and Fig.~\ref{fig.SM_time} shows how the resulting values of $\xi(t^*)$ change as $t^*$ and $d$ are varied. At short times there is a visible change in the length scale $\xi(t)$, however, at longer times we see that it appears to saturate towards a well-defined length scale with only a weak dependence on time. It is possible that simulations which extend to longer times may be able to improve upon the results presented here, but our results suggest this will be a meagre quantitative improvement in exchange for a great deal of computational effort.

\begin{figure}[ht]
	\centering
	\includegraphics[width=\linewidth]{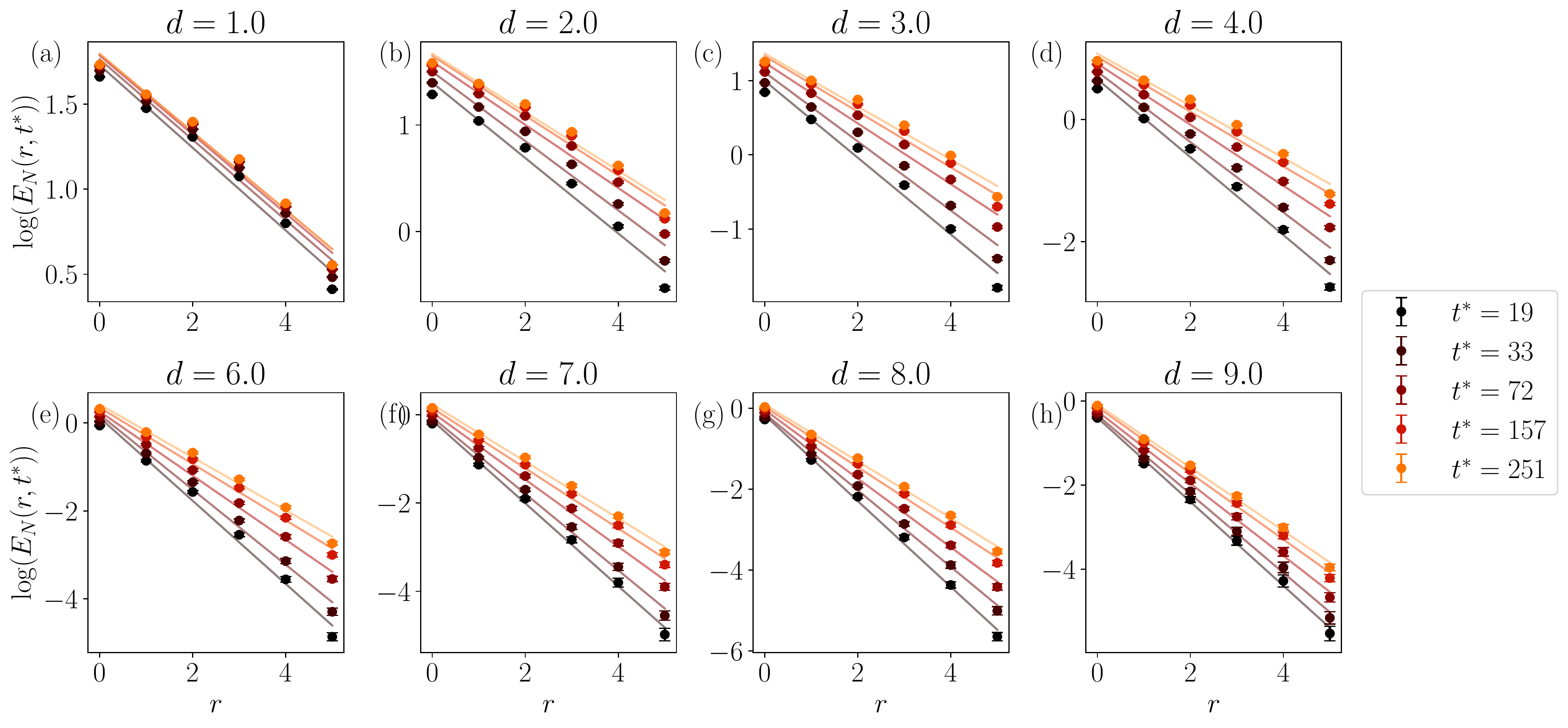}
	\caption{\textbf{Fits of the spatial decay of entanglement negativity for different times and disorder strengths: }A comparison of the fits to the entanglement negativity $E_N(r,t^*)$ at different times $t^*$ and for different values of the disorder strength $d$, shown for $L=14$ and averaged over $N_s=240$ disorder realisations. The circular markers show the data, while the solid lines indicate the linear fits used to extract the $l$-bit localisation length. Error bars showing the standard error in the mean are typically smaller than the marker size.
		At large disorder strengths (i.e., in the localised phase), the linear fit is very good, confirming that the negativity does indeed decrease exponentially with distance in this phase. After a time of $t^* \approx 100$, the gradient of the decay (and hence the corresponding $l$-bit localisation length) does 
		not strongly change with time in the localised phase.}
	\label{fig.SM_fits}
\end{figure}

\begin{figure}[ht]
	\centering
	\includegraphics[width=\linewidth]{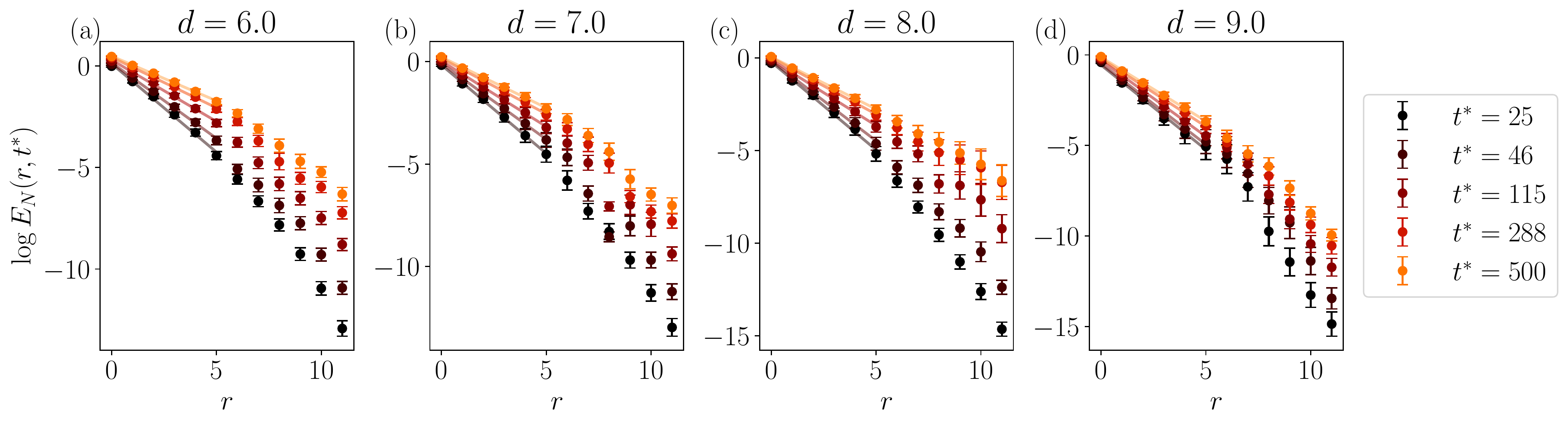}
	\caption{\textbf{Fits of the spatial decay of entanglement negativity for different times and disorder strengths: }The same as in Fig.~\ref{fig.SM_fits}, but for $L=24$, $\chi=192$ and averaged over $N_s=100$ disorder realisations, shown only in the localised phase. The solid lines indicate the range of points over which the exponential fits were performed in order to extract the length scale shown in Fig.~3 of the main text.}
	\label{fig.SM_fits2}
\end{figure}

\begin{figure}[ht]
	\centering
	\includegraphics[width=0.75\linewidth]{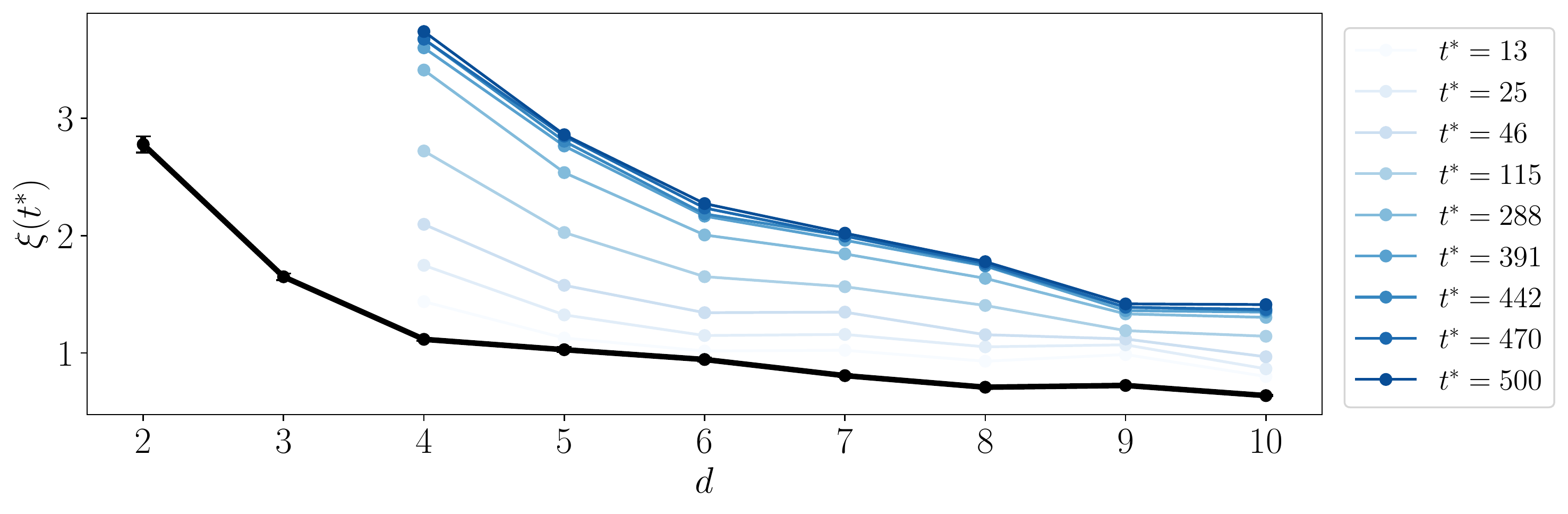}
	\caption{\textbf{Extracted length scales at different times: }A comparison of how the length scale $\xi(t^*)$ changes as the measurement time $t^*$ is varied. The results here are extracted from the fits shown in Fig.~\ref{fig.SM_fits2}. At short times, the value of $\xi(t^*)$ changes rapidly, however, at longer times, the dependence of $t^*$ weakens significantly. The black line is the Anderson localisation length for comparison, as discussed in the main text. For clarity, error bars are not shown except on the Anderson localisation length. (Note that error bars are shown in the data presented in the main text.)}
	\label{fig.SM_time}
\end{figure}

\section{Supplementary note 5: Negativity between subsystems of fixed size}

It is also possible to compute the entanglement negativity on a more general interval between two subsystems of fixed size $R_b$ separated by a distance $r$, as sketched in Fig.~\ref{fig.SM_disjoint}. In this case, a very similar procedure to that proposed in the main text is possible, with the caveat that one must carefully choose both the block size $R_b$ and the distance $r$ to ensure that the extraction of the $l$-bit length scale is done during the dephasing regime. 

To be specific, if the block size $R_b$ is small and the subsystems are close together (i.e., $r$ is small), then the entanglement negativity will rapidly saturate (as the maximum entanglement is controlled by the size of the subsystems under consideration). On the other hand, if the blocks are widely separated (i.e., large $r$), then the negativity will remain zero until times exponentially large in $r$. If one were to consider a small subsystem of $R_b = 2$, for example, then the entanglement negativity for small values of $r$ would saturate well before widely separated subsystems have time to become entangled, this meaning that the data points for both small $r$ and large $r$ would have to be discarded when performing the fit. This can be ameliorated by using blocks of intermediate size, such that they are large enough that the subsystem entanglement does not saturate too rapidly and the data at small values of $r$ remains reliable. Representative results for this case are shown in Fig.~\ref{fig.SM_disjoint2}, where it is clear that for block size $R_b=2$, the curves for small values of $r$ saturate too quickly to be used in extracting the $l$-bit length scale. Block size $R_b=3$ is better, and here $R_b=4$ offers the best compromise, with a clearly identifiable region where the curves for all values of $1 \leq r < 5$ are in a regime of logarithmic growth with approximately the same rate, and a length scale may be extracted. On the other hand, for block size $R_b=5$, with a system size of $L=16$ it is not possible to separate the blocks widely enough to extract enough data to perform a reliable fit. (Note that in these computations, we avoid subsystems that contain the two sites on each end of the chain in order to reduce finite-size effects. In addition, we average over all possible positions of the blocks with size $R_b$ separated by a distance $r$ within our system of size $L$.)

\begin{figure}[h]
	\centering
	\includegraphics[width=0.4\linewidth]{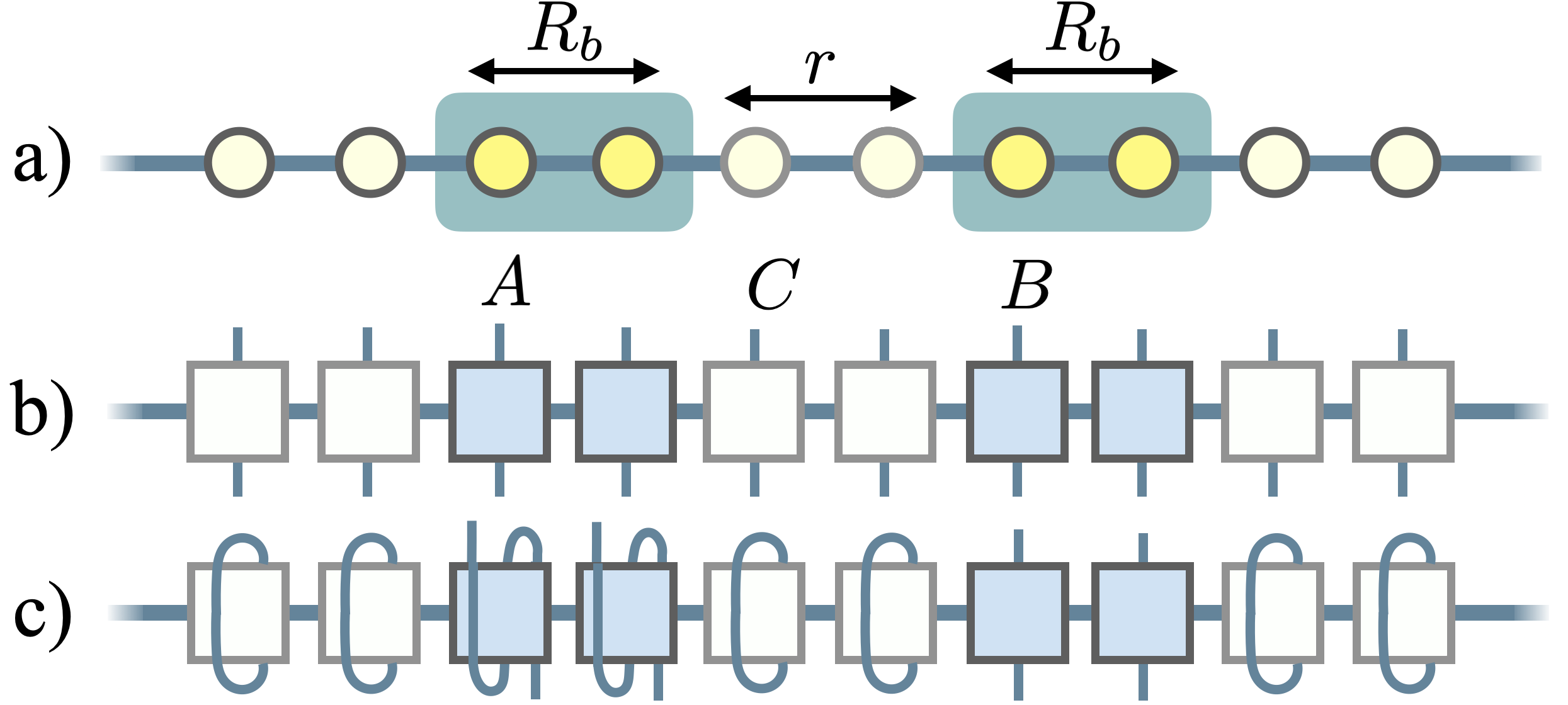}
	\caption{\textbf{Division into subsystems and computation of negativity: }A sketch of how the entanglement negativity can be used to quantify the entanglement between subsystems of fixed size $R_b$ separated by a distance $r$. a) A sketch of the spin chain, identifying the subsystems $A$ and $B$ and the relevant distances $R_b$ and $r$. b) A sketch of the corresponding density matrix in MPO form, identifying the subsystems. c) A sketch of how the negativity is computed in this case, tracing out the complement of $A$ and $B$ while still applying the same `twist' to the MPO legs in order to compute the partial transpose, as in Fig.~1 of the main text.}
	\label{fig.SM_disjoint}
\end{figure}

\begin{figure}[h]
	\centering
	\includegraphics[width=\linewidth]{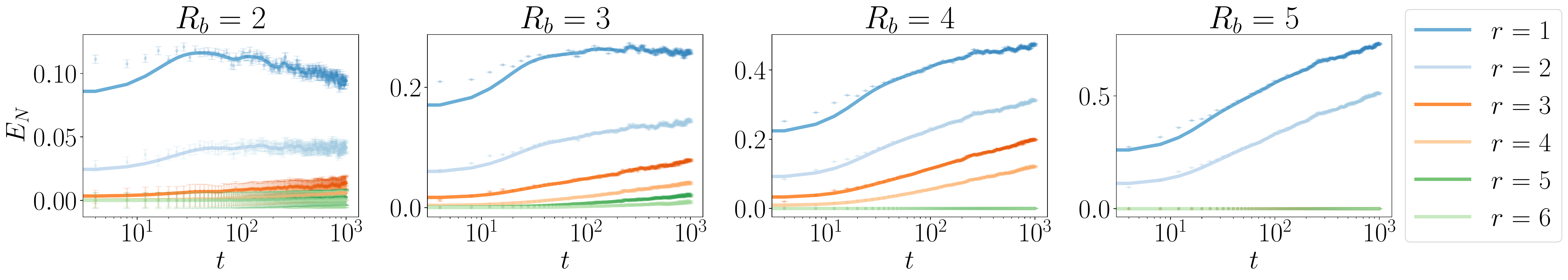}
	\caption{\textbf{Growth of negativity between subsystems of fixed size: }Entanglement negativity following a quench from a Néel state, shown for disorder strength $d=6.0$ with system size $L=16$, averaged over $N_s=96$ disorder realisations. Here we compute the negativity for a subsystem of fixed size $R_b:=|A|=|B|$, as sketched in Fig.~\ref{fig.SM_disjoint}. We can see that for small values of $R_b$, the negativity at small separations $r$ saturates quickly and that only for larger values of $R_b$ does the negativity increase in a manner that allows extraction of the relevant $l$-bit length scale. Error bars show the standard error in the mean, and the solid lines are a smoothed guide to the eye.}
	\label{fig.SM_disjoint2}
\end{figure}

\section{Supplementary note 6: Mutual information}

In the main text and in our analytical results, we specified the logarithmic negativity as our chosen spatially-resolved entanglement probe. Here we briefly demonstrate that another spatially-resolved entanglement measure, the \emph{mutual information}, also exhibits qualitatively similar behaviour.
The \emph{mutual information} between subsystems $A$ and $B$ is defined as
\begin{align}
	I(A:B) := S(\rho^A) + S(\rho^B) - S(\rho^{AB})
\end{align}
where $S(\rho) := -\tr \rho \log(\rho)$ represents the \emph{von Neumann entropy} 
of a quantum state $\rho$, and $\rho^A$ is the reduced quantum state of subsystem $A$. Fig.~\ref{fig.SM_MI} shows the results for a small system of size $L=12$ with bond dimension $\chi=128$, averaged over $N_s=100$ disorder realisations. The mutual information is qualitatively -- and even quantitatively, in many cases -- similar to the entanglement negativity, strongly suggesting that it would be a more than acceptable substitute and that much of the intuition developed in the main text should also apply to the mutual information. These findings are compatible with those using the
multi-partite mutual information as a probe for many-body localization
\cite{gooldTotalCorrelationsDiagonal2015,detomasiQuantumMutualInformation2017}.

\begin{figure}[h]
	\centering
	\includegraphics[width=\linewidth]{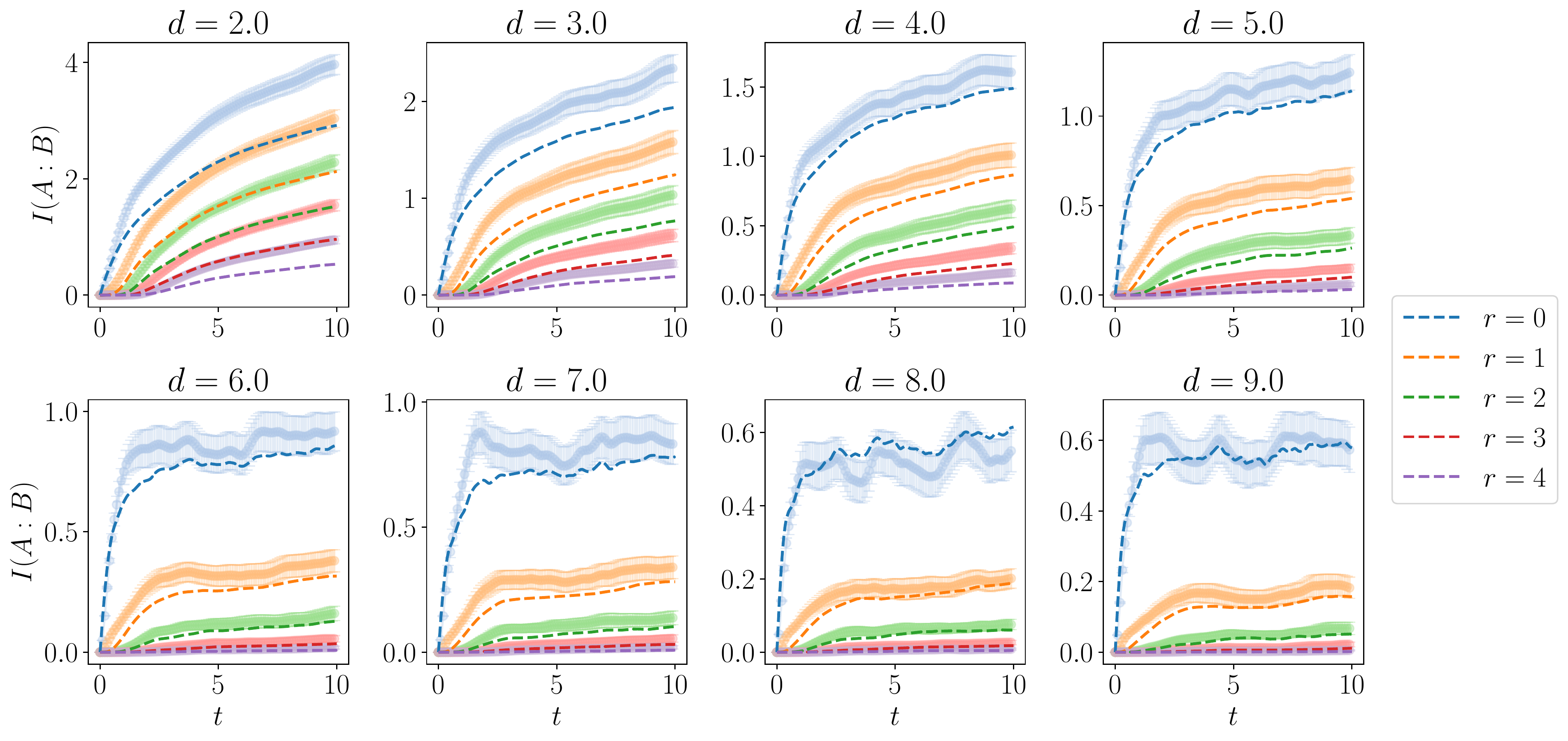}
	\caption{\textbf{Mutual information between separated regions over time: }A comparison of mutual information (solid lines) and entanglement negativity (dashed lines) between subsystems $A$ and $B$ separated by a distance $r$, for a system of size $L=12$ with bond dimension $\chi=128$, averaged over $N_s=100$ disorder realisations.}
	\label{fig.SM_MI}
\end{figure}
\section{Supplementary note 7: Computation of the $l$-bits and the quasi-locality measure}

The $l$-bits in this work are calculated using the method described in 
Ref.~\cite{goihlConstructionExactConstants2018}. For 
completeness, we offer a brief explanation of the method:
A system in the fully many-body localised phase can be fully characterized by a 
complete set of $l$-bits. Let $H$ be the MBL Hamiltonian of a system of size $N$ and $\tau_i,\ i=1,\dots,N$, 
be the $l$-bits, then they need to satisfy the following properties:
\begin{enumerate}
	\item $[H,\tau^z_i] = 0$,
	\item $[\tau^z_i,\tau^z_j] = 0$ for all $i$ and $j$,
	\item $\tau_i$ is quasi-local in real space (will be elaborated in the next section).
\end{enumerate}
A simple prescription in Ref.~\cite{kimLocalIntegralsMotion2014} has been given to construct the $l$-bits out of infinite time averages of terms that the Hamiltonian contains. The infinite-time average of a term $h_j$ is
\begin{equation}\label{eqn: inf_time_avg}
	\mathbbm{E}(h_j) = \lim_{T\rightarrow\infty}\frac{1}{T}\int_0^T e^{-iHt}\ h_j\ e^{iHt} \dd t = \sum_k \bra{E_k}h_j\ket{E_k} \ket{E_k}\bra{E_k},
\end{equation}
assuming a non-degenerate spectrum $\{E_k \}$ of $H$. Therefore, $\mathbbm{E}(h_j)$, being sums of projectors onto the eigenstates of $H$, automatically fulfill properties 1 and 2. The authors have also demonstrated the quasi-locality of the resulting operator, showing that the non-local contributions from the off-diagonal elements, i.e., contributions from $\tau^x$ and $\tau^y$ will be removed through the infinite-time averaging procedure. The downside with this approach is that the resulting operator has a degenerate spectrum that is distinct from what the Pauli algebra mandates. Thus, we can no longer associate this operator with the picture of having ladder operators that help us transverse through the different modes/eigenstates. 

The authors in Ref.~\cite{goihlConstructionExactConstants2018} have computed $\mathbb{E}(\sigma_j^z)$ from (\ref{eqn: inf_time_avg}) and re-arranged the order of the eigenstates in $U_d$ to minimize the pair-wise differences between the spectrum of $\mathbb{E}(\sigma_j^z)$ and that of 
\begin{equation}
	\tau^z_j = U^{\dagger}_d \sigma^z_j U_d, 
\end{equation}
starting with $j=1$ and sequentially optimising the eigenstate order with larger $j$ while keeping the already optimized partial ordering from smaller $j$ intact. The resulting operators will preserve the Pauli algebra by construction while becoming quasi-local in real space. For simplicity, we only deal with trace-less and Hermitian operators for the moment. Operators with non-vanishing traces require special procedures to meet the orthonormality of the Hilbert-Schmidt inner product. 

\begin{definition}[Quasi-locality]\label{def: quasi_local}
	Let $\tau$ be a trace-less Hermitian operator,  normalized with respect to the Frobenius norm, and $L_i$ its orthogonal trace-less Hermitian basis. Consider the decomposition
	\begin{equation}
		\tau = \sum_i a_i L_i.
	\end{equation} 
	Let $S(L_i)$ be the support of $L_i$ in real-space. $\tau$ is quasi-local around site $j$, if and only if for any connected region $B$ containing $j$, we have 
	\begin{equation}
		\sum_{S(L_i) \nsubseteq B} \abs{a_i}^2 \leq k \exp(-\text{dist}(j,B)/\xi).
	\end{equation}
\end{definition}

This definition, although rigorous, can be cumbersome, and it is not easy to compute
the relevant quantities. 
Let us consider the partial trace of the operator $\tau$,
\begin{equation}
	\tr_{B^C} (\tau) =  2^{\abs{B^C}}\sum_{S(L_i) \subseteq B} a_i \Tilde{L}_i,
\end{equation}
where $\Tilde{L}_i$ are now defined on a smaller Hilbert space contained in $B$ and the extra identity operators outside $B$ become the pre-factor $2^{\abs{B^C}}$ after the partial trace. Then, we can compute the square of the Frobenius norm of this truncated operator to get
\begin{equation}
	\norm{\tr_{B^C} (\tau)}^2_2 = 2^{2\abs{B^C}}\tr(\sum_{S(L_i) \subseteq B} a^*_i \Tilde{L}_i \sum_{S(L_j) \subseteq B} a_j \Tilde{L}_j) = 2^{2\abs{B^C}+\abs{B}}  \sum_{S(L_i) \subseteq B}  \abs{a_i}^2,
\end{equation}
due to the orthogonality of the basis operators $L_i$ and their restrictions to any region $B$. 
Finally, we leverage 
the normalisation of $\tau$ to observe that
\begin{equation}
	\sum_{S(L_i) \subseteq B}  \abs{a_i}^2 + \sum_{S(L_i) \nsubseteq B}  \abs{a_i}^2 = \norm{\tau}^2_2 = 1.
\end{equation}
Therefore, 
\begin{equation}\label{def: quasi_local_partial}
	\sum_{S(L_i) \nsubseteq B}  \abs{a_i}^2 = 1-\frac{1}{2^{\abs{B} +2\abs{B^C}}}\norm{\tr_{B^C}(\tau)}_{2}^2  \leq k \exp(-\text{dist}(j,B)/\xi).
\end{equation}
This is exactly the quasi-locality measure proposed in Refs.~\cite{goihlConstructionExactConstants2018, Chandran+15}. 
The spatial decay of the $l$-bits can also be computed using the weight measure in Eq.~(6) from Ref.~\cite{kulshreshthaBehaviorLbitsManybody2018}. This notion of quasi-locality is 
agnostic to operator content and where it is quasi-local around because one only 
needs to compute the weight of the operator filtered at each site. 

\section{Supplementary note 8: Entanglement growth bound}

In this section, we prove an entanglement growth bound that has 
been tailored for the case at hand. Generic 
entanglement growth bounds for the entanglement entropy 
have been proven for local Hamiltonians \cite{EntanglementRates,PhysRevLett.97.150404,PhysRevLett.127.020501,kimLocalIntegralsMotion2014}. 
Here, we prove novel bounds for the negativity, ones
that are tailored to be applicable
to many-body localised systems and the 
geometry at hand. We now look at the specific setting of having two regions, $A$ and $B$ 
separated by a 
region $C$ that has been traced out, and ask how much 
entanglement can be generated by a 
many-body localised Hamiltonian. 
We consider states evolving in time as
$\rho(t):= e^{-itH}\rho e^{itH}$.
For the purpose of this proof, we will consider a slightly strengthened definition of quasi-locality with respect to the definition in the main text, namely, instead of the normalized Frobenius norm, we will assume the $l$-bits are localised in the operator norm, i.e., the following.

\begin{definition}[Strong quasi-locality]\label{def:quasilocality_strong}
	An operator is said to be quasi-local around a region $R$, if for any region $X\subset R$,
	it  holds that
	\begin{equation}
		\left\|O- \frac{1}{2^{|X^c|}}\tr_{X^c}(O)\otimes \mathbb{I}_{X_c}\right\|_{\infty}^2\leq \|O\|_{\infty}^2Ke^{-d(R,X^c)/\xi}
		\label{eq.quasi-local_inf}
	\end{equation}
	for some constant $K>0$.
\end{definition}

This definition is stronger than the definition in the main text, in the sense that if an operator is quasi-local in the sense above, it is also quasi-local in the sense of the main text. For convenience, we repeat here the definition of a many-body localized Hamiltonian:
\begin{definition}[Many-body localisation]\label{def:mbl_hamiltonian_sm} A Hamiltonian 
	\begin{equation}
		H= \sum_{j=1}^n \omega_j^{(1)} h_j + \sum_{j,k=1}^n \omega^{(2)}_{j,k} h_j h_k + \dots
	\end{equation}
	with real weights $\{\omega_j^{(1)}\}$ and $\{\omega_{j,k}^{(2)}\}$,
	is called many-body localised if it can be written as a sum of mutually commuting ($[h_j,h_k]=0$
	for all $j,k$) quasi-local terms $h_j$, each centred around site $j$, and if $\omega_{i_1,\dots, i_n}\leq \omega e^{-|i_1-i_n|/\kappa}$, where $i_1< i_2\dots < i_n$.
\end{definition}

The following is the central statement from which the entanglement bound easily follows:
\begin{theorem}[Entanglement growth bound for sums of quasi-local operators] 
	\label{thm:ent_growth_sum}
	Let $\rho$ be any initial state. Let $H$ be a many-body localised Hamiltonian as per Definition \ref{def:mbl_hamiltonian_sm} with localisation length $\xi\leq {1}/{(4\log(2))}$ and $2(1/\kappa-\log(2))> 1/\xi$, consider three blocks $A,C,B$ such that $C$ divides $A$ from $B$ the
	growth of the  negativity of the state
	$\rho(t) = e^{-itH}\rho e^{it H}$ restricted to the regions $A,B$ for times $t\geq 0$ and for any $r\geq |C|/2$
	\begin{equation}
		E_N(t)=\log_2(\|\rho_{A,B}^{T_A}(t)\|_{1}) \leq 4r-2|C|+ tO(e^{-r/(2\xi)}).
	\end{equation}
\end{theorem}
Notice the the $r$ in the above is a parameter one can choose freely.
We obtain the following
statement, morally, by picking $r=2\xi\log(t)$ in the above bound if $t> e^{|C|/(4\xi)}$ and $r=|C|/2$ otherwise. Some additional details are required in order not to consider a time dependent $r$ in the proof of Theorem \ref{thm:ent_growth_sum}.

\begin{corollary}\label{cor:final_bound}[Logarithmic growth of the negativity]
	If $t\geq e^{|C|/{4\xi}}$, under the assumptions of Theorem \ref{thm:ent_growth_sum}, we have
	\begin{equation}
		E_N(t)\leq \min\{t\,O(e^{-|C|/(4\xi)}), 8\xi\log(t)-2|C|\}+
		O(1),
	\end{equation}
	while for $t<e^{|C|/{4\xi}}$,
	\begin{equation}
		E_N(t)\leq t\,O(e^{-|C|/(4\xi)})
	\end{equation}
\end{corollary}
\begin{proof}
	For the $t<e^{|C|/{4\xi}}$ case, one need only pick $r=|C|/2$. Otherwise, for any time $t$, let $r^*$ be the positive integer such that $e^{r^*/(2\xi)}\leq t \leq e^{(r^*+1)/(2\xi)}$. Then $r^*\leq 2\xi\log(t)\leq r^*+1$. Picking $r=r^*$ we get 
	\begin{equation}
		E_N(t)= \leq 4r^*-2|C|+ tO(e^{-r^*/(2\xi)})\leq 8\xi\log(t)-2|C|+ O(t\,e^{\log(t)+1/(2\xi)})=8\xi\log(t)-2|C|+O(1),
	\end{equation}
	since the bound with $r=|C|/2$ still holds, one can pick the minimum between these two bounds.
\end{proof}
To turn this result into a proper bound for many-body localised Hamiltonians, we need to reduce such Hamiltonians to the form considered in Theorem \ref{thm:ent_growth_sum}. For this purpose, we will assume a slightly stronger definition of the many-body localised Hamiltonian. We call a unitary quasi-local with localisation length $\xi$ if it maps any local operator on a region $R$ to a quasi-local operator around $R$ with localisation length $\xi$. We will assume that the Hamiltonian is diagonalised by a quasi-local unitary, which means that the $l$-bits are simply dressed Pauli $Z$ operators, $h_i=U\sigma_z^i U^{\dagger}$. In particular this implies that a product of $l$-bits $h_{i_1}h_{i_2}\dots h_{i_n}$, with $i_1<i_2<\dots < i_n$ is quasi-local around $\{i_1,\dots, i_n\}$. 
For the purposes of the subsequent discussion, it will be useful to define the projector 
\begin{equation}
	\mathcal P_X(O)= \tr_{X^c}(O)\otimes \frac{I}{2^{|X^c|}},
\end{equation}
for any region $X$ of the lattice. It is easy to verify that $\mathcal P_X$ is a projector and that it is self-adjoint in the Hilbert Schmidt inner product.  In what follows, $\|\cdot\|$ denotes the operator norm.

\begin{lemma}[Quasi-local sums]\label{lem:quasi-local-sums}
	Let $H$ be a many-body localised Hamiltonian in the sense described above with localisation length $\xi$, and in addition, assume $2(1/\kappa-\log(2))> 1/\xi$, where $\kappa$ has been defined in Definition \ref{def:mbl_hamiltonian_sm}. Then the Hamiltonian can be written as %\blu{I think it's helpful to remind the reader about how $\kappa$ is defined in definition 2 in the main text. }
	\begin{equation}
		H=\sum_{i=1}^N H_i
	\end{equation}
	where $H_i$ is quasi-local around the site $i$ with localisation length $2\xi$.
\end{lemma}
% \je{[I have reformulated this, please check.]}
We note that, a quasi-local operator with localisation length $\xi$ is quasi-local for any localisation length $\xi'>\xi$, if 
$2(1/\kappa-\log(2))\leq  1/\xi$, it suffices to relax the quasi-locality of the $l$-bits by increasing the localisation length until this condition is satisfied, i.e., 
choose
\begin{equation}
	\xi'=\frac{1}{2(1/\kappa-\log(2))}+\epsilon >\xi
\end{equation}
for some constant $\epsilon>0$. 
This will result in faster ($\sim r/\xi')$, but still exponentially slow, growth of entanglement in the bound.% \blu{this argument is a bit hard to understand for me}
\begin{proof}
	Define $H_i$ as
	\begin{equation}
		H_i=\omega_i h_i + \sum_{l\geq 1} \sum_{k=2}^{2l+1}\sum_{I_k\in B_k(i,l)} \omega_{I_k}h_{I_k}
	\end{equation}
	where $I_k=\{i_1,\dots, i_k\}$ 
	are collections of $k$ sites contained in $i-l\dots i+l$, containing the left boundary and either the right boundary or the site immediately to the left of it. We denote the set of all such $I_k$ as $B_k(i,l)$ (c.f. Figure \ref{fig:lbit_grouping}).
	\begin{figure}[h]
		\centering
		\includegraphics[width=.3\textwidth]{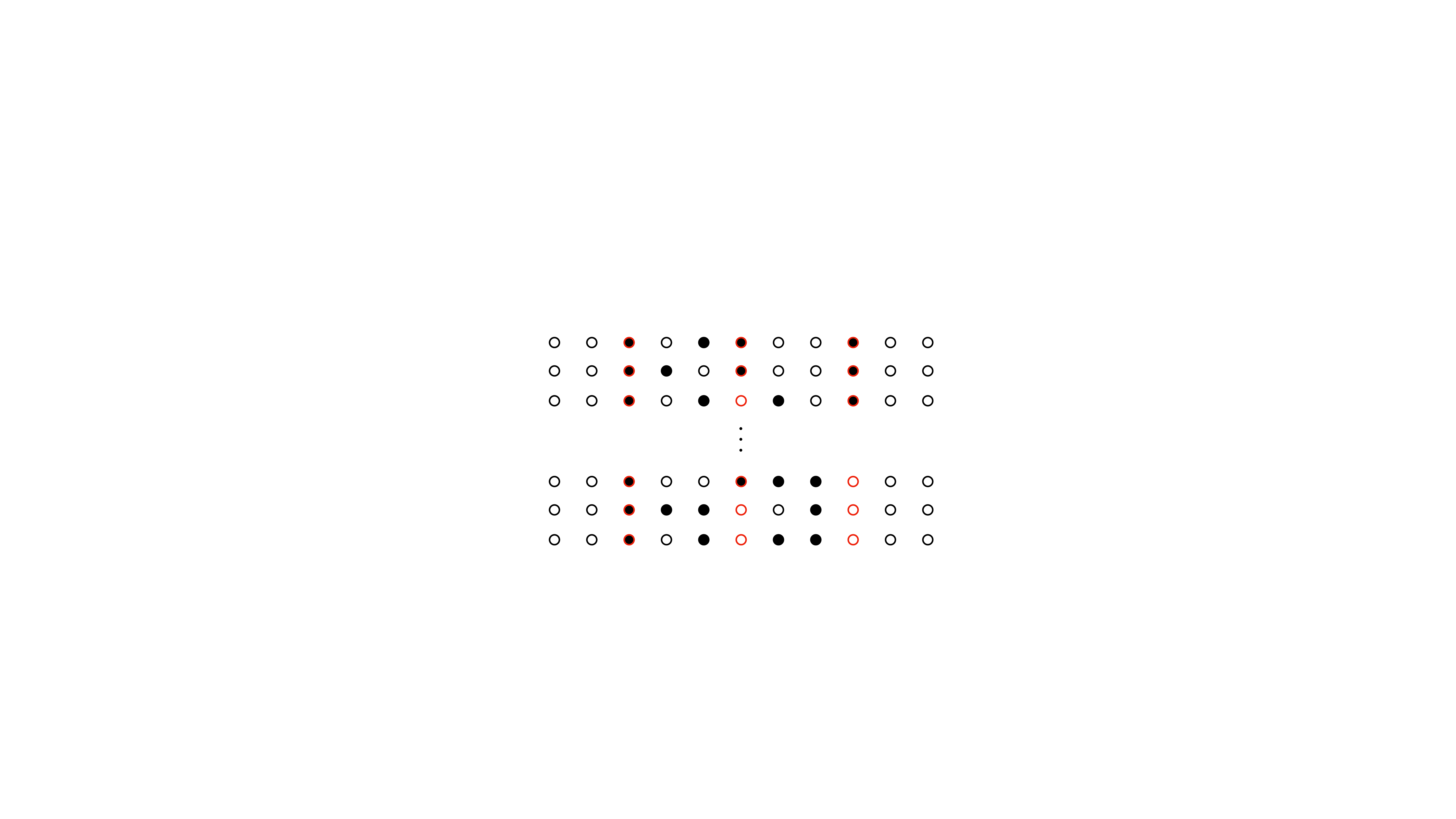}
		\caption{\textbf{Examples of sets of sites contained in $B_k(i,l)$: }In this example, $k=4$ and $l=3$. The sites $i$, $i+l$, $i-l$ are highlighted in red.}
		\label{fig:lbit_grouping}
	\end{figure}In the above, we have defined $h_{I_k}=h_{i_1}\dots h_{i_k}$. We have then $\omega_{I_k}^2\leq \omega^2 e^{-(2l-1)/\kappa}$ and that $h_{I_k}$ is quasi-local around $[i-l,i+l]$ with localisation length $\xi$. We have 
	$H=\sum_i H_i$, it remains to show that the $H_i$ are quasi-local as claimed. Let $R=[i-r,i+r]$ be a stretch of sites of radius $r$ around $i$. We have, by the triangle inequality,
	\begin{equation}
		\begin{aligned}
			\|H_i-\mathcal P_R(H_i)\| &= \omega_i \|h_i-\mathcal P_R(h_i)\| + \sum_{l\geq 1} \sum_{k=2}^{2l+1}\sum_{I_k\in B_k(i,l)} \omega_{I_k} \|h_i-\mathcal P_R(h_i)\|\\& \leq \omega K e^{-r/\xi} +K \sum_{l=1}^{r-1}\sum_{I_k\in B_k(i,l)} \omega e^{-(2l-1)/\kappa}e^{-(r-l)/(2\xi)} + 2\sum_{l\geq r} \sum_{I_k\in B_k(i,l)}\omega e^{-(2l-1)/\kappa},
		\end{aligned}
	\end{equation}
	where we have used quasi-locality and in the $l\geq r$ part of the sum, that $\|h_i-\mathcal P_R(h_i)\|\leq 2$.
	Now notice that
	\begin{equation}
		\sum_{I_k\in B_k(i,l)}=2\cdot 2^{2l-2}=2^{2l-1} \leq 2^{2l} ,
	\end{equation}
	we then have that
	\begin{equation}
		\begin{aligned}
			\|H_i-\mathcal P_R(H_i)\| &\leq \omega K e^{-r/(2\xi)} +K \sum_{l=1}^{r-1}2^{2l} \omega e^{-2l/\kappa}e^{-(r-l)/(2\xi)} + 2\sum_{l\geq r} 2^{2l}\omega e^{-2l/\kappa}\\&=\omega K e^{-r/(2\xi)} +K\sum_{l=1}^{r-1} \omega e^{-2l(1/\kappa-\log(2))}e^{-(r-l)/(2\xi)} + 2\omega \sum_{l\geq r} e^{-2l(1/\kappa-\log(2))}.
		\end{aligned}
	\end{equation}
	The first term is already bounded as needed, and for the third term, we use $2(1/\kappa-\log(2)) > 1/(2\xi)$, to get 
	\begin{equation}
		\sum_{l\geq r} e^{-2l(1/\kappa-\log(2))}\leq \sum_{l\geq r} e^{-l/(2\xi)}= \frac{1}{1-e^{-1/(2\xi)}}e^{-r/(2\xi)}.
	\end{equation}
	For the second term, let $\delta= 2(1/\kappa-\log(2))-1/(2\xi) >0$. Then 
	\begin{equation}
		\sum_{l=1}^{r-1} e^{-2l(1/\kappa-\log(2))}e^{-(r-l)/(2\xi)} =e^{-r/(2\xi)} \sum_{l=1}^{r-1} e^{-\delta l}\leq  \frac{e^{-\delta}}{1-e^{-\delta}}e^{-r/(2\xi)}.
	\end{equation}
	We note in passing that if $\delta=0$, the bound does not diverge, but an additional linear term is added to get a bound of the form $re^{-r\xi}$. In conclusion,
	\begin{equation}
		\|H_i-\mathcal P_R(H_i)\|\leq K' e^{-r/(2\xi)}
	\end{equation}
	with 
	\begin{equation}
		K':=\omega\left(K+K\frac{e^{-\delta}}{1-e^{-\delta}}+\frac{2}{1-e^{-1/\xi}}\right),
	\end{equation}
	that is,
	\begin{equation}
		\|H_i-\mathcal P_R(H_i)\|^2\leq K'^2 e^{-r/\xi}
	\end{equation}
	which is the definition of quasi-locality noticing $||h_i||=1$.
\end{proof}

The rest of this section is dedicated to proving Theorem \ref{thm:ent_growth_sum}. The technique used for the proof is analogous to that of Ref.~\cite{kimLocalIntegralsMotion2014}, where an analogous bound was derived for the entanglement entropy in the case $|C|=0$, the properties of the negativity, which is a meaningful entanglement measure for mixed states, allow the generalisation to disconnected regions. The following two properties of the $1$-norm and the partial transpose will be used repeatedly:
\begin{itemize}
	\item  The $1-$norm contracts under the partial trace, 
	i.e., $||\rho_A||_1\leq ||\rho_{A,B}||_1$ for 
	any bipartite state $\rho_{A,B}$ \cite{Rastegin_2012},
	\item $||\rho_{A,B}^{T_B}||_1\leq d_{B}||\rho_{A,B}||_1$ \cite{Ando2008}.
\end{itemize}

We will now prove Theorem \ref{thm:ent_growth_sum}.
From now on we will assume that $|C|$ is even, the odd case is analogous. We will label the two central sites of $C$ as $\pm 1$, all the sites to the right of the center of $C$ have positive integer labels, and to the left negative integer labels (notice that there is no site $0$). In addition, we will use the standard notation $[a,b)$ denoting real intervals to denote intervals in the chain, so from now on $[a,b)$ is understood to mean $[a,b)\cap\mathbb Z/\{0\}$.

Divide the chain into the following regions (see Fig.~\ref{fig:regions}), recall that $r$ is an arbitrary integer greater than $|C|/2$.
\begin{align}
	&\bar L=(-\infty, -r+|C|/2), \quad L=(-\infty,|C|/2),\\
	&\bar \del=[-r+|C|/2,r-|C|/2], \quad \del=[-2r+|C|/2, 2r-|C|/2],\\
	&\bar R=(r-|C|/2, \infty), \quad R=(-|C|/2,\infty),
\end{align}
and divide the terms in the Hamiltonian accordingly as
\begin{equation}
	H=h_{\bar L}+h_{\bar R}+h_{\bar \del}
\end{equation}
where $h_{\bar S}=\sum_{i\in S} h_i$ for $S\in{L,\del,R}$. We will now split the term $h_{\bar S}$ into a main local part acting on $S$ and a tail whose norm is exponentially small in $r$.
\begin{figure}[h!]
	\centering
	\includegraphics[width=.7\textwidth]{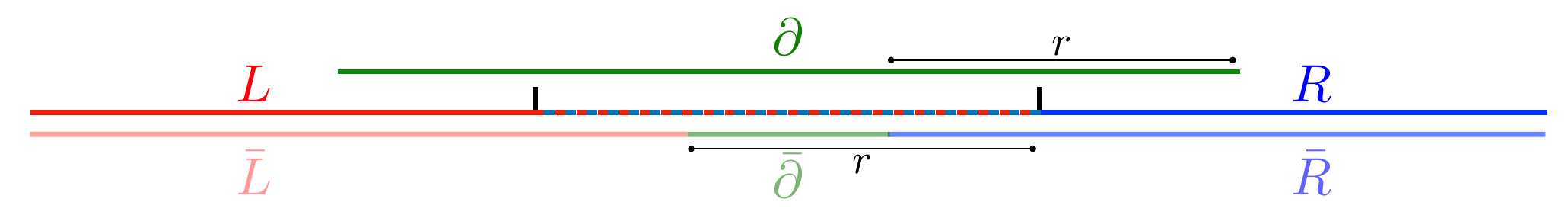}
	\caption{\textbf{Division of the chain into regions for the proof of Theorem \ref{thm:ent_growth_sum}:} The central region bounded by the vertical lines is $C$, represented by the dashed red and blue line. Note that $C$ is contained in both $L$ and $R$. }
	\label{fig:regions}
\end{figure}

More specifically, we can isolate the tail evolution as follows
\begin{lemma}
	Let $H_S=\mathcal P_S(h_{\bar S})$ and define the corresponding tail as $T_S=h_{\bar{S}}-H_S$, then 
	\begin{equation}
		e^{ih_{\bar S}t}=e^{i H_St}U^T_S(t)
	\end{equation}
	where the tail evolution is given by
	
	\begin{equation}
		U_S^T(t)=\mathcal T\left(\exp\left(\int_{0}^t \mathrm d\tau \, \tilde T_S(\tau)\right)\right)=\sum_{k=0}^{\infty} i^k\int_0^t\mathrm dt_1\,\dots \int_0^{t_{k-1}} \mathrm{d}t_{k} \tilde T_S(t_1)\dots \tilde T_S(t_k)
	\end{equation}
	where 
	\begin{equation}
		\tilde T_S(\tau)=e^{-iH_S\tau}T_Se^{iH_S\tau}.
	\end{equation}
\end{lemma}
\begin{proof}
	We have $U^T_S(t)=e^{-iH_St}e^{ih_St}$, hence
	\begin{equation}
		\frac{\mathrm{d}}{\mathrm{d}t}U^T_S(t)= -iH_Se^{-iH_St}e^{ih_{\bar{S}}t}+ie^{-iH_St}h_Se^{ih_{\bar{S}}t}= ie^{-iH_St}Te^{ih_{\bar{S}}t}=\tilde T_S(t) U^T_S(t)
	\end{equation}
	this is the differential equation satisfied by the time evolution operator with the time-dependent Hamiltonian $\tilde T_S(t)$, it is well known that its solution is given by the Dyson series in the statement.
\end{proof}
Notice that the tail evolution is an operator acting globally on the lattice, hence it could still, in principle, create an amount of entanglement extensive in the system size. The intuition is that it is generated by $\tilde T_S$, which has exponentially small norm. We need to formalize the idea that this leftover tail evolution cannot create much entanglement.
For each $S$, define $S+q$ as the region $S$ extended to the left and the right (if possible) by $q$ sites. We can then write the tail $T_S$ as a telescopic sum, defining $H_{S+q}= \mathcal P_{S+q}(h_{\bar S})$
\begin{equation}
	T_S= -H_S+H_{S+1}-H_{S+1}+H_{S+2}-H_{S+2}\dots= -\sum_{q=0}^{\infty} H_{S+q}-H_{S+q+1}.
\end{equation}
Clearly $h_{\bar S}= \lim_{q\to \infty} \mathcal P_{S+q}(h_{\bar S})$, as this is simply the ``truncation" where the whole lattice is kept.
We then have that 
\begin{equation}
	T_S=\sum_{q=0}^{\infty} Z_{S+q+1}
\end{equation}
where $Z_{S+q}:= H_{S+q}-H_{S+q+1}$ acts only on $S+q+1$. Until now, we have claimed that these tail operators have exponentially small norms in $r$. Before continuing, let us prove this.

\begin{lemma}[Tail bounds]\label{lem:tail}
	$||Z_{S+q}||\leq O(e^{-(r+q)/(2\xi)})$.
\end{lemma}
\begin{proof}
	We will prove that $||h_{\bar{S}}-H_{S+q}||\leq O(e^{-(r+q)/(2\xi)})$, the statement follows as
	\begin{equation}
		||Z_{S+q}||\leq  ||H_{S+q}-H_{S+q+1}||\leq (||H_{S+q}-h_{\bar S}|| + ||H_{S+q+1}-h_{\bar S}||) .
	\end{equation}
	Let us start with $S=\del$. The closest $l$-bits to the boundary of $\del+q$ in $\bar \del$ are separated from it by $r+q$ sites, and the furthest ones by $2r+q$. The sum is clearly symmetric around the centre of the chain, and we use that for any $S$, $S$ is constructed such that the distance from any site in $\bar S$ to the edge of $\bar S$ is at least $r$. Recall that $h_{\bar S}=\sum_{i\in \bar S}h_i$, where the $h_i$ are quasi-local as in Definition \ref{def:quasilocality_strong}, so that 
	\begin{equation}
		||h_{\bar{\del}}-H_{\del+q}||\leq \sum_{k=-r+|C|/2}^{r-|C|/2}||h_k-\mathcal{P}_{\del+q} (h_k)||\leq 2K\sum_{j=r+q}^{2r+q-|C|/2} e^{-j/(2\xi)}\leq O(e^{-(r+q)/(2\xi)})
	\end{equation}
	for some constant $K>0$, 
	by the definition of quasi-locality. Let us move to $S=R$, where again the closest l-bits in $\bar R$ to the boundary of $R+q$ are separated from it by $r+q$ sites, and the furthest has a distance unbounded in the system size. We have
	\begin{equation}
		||h_{\bar{R}}-H_{R+q}||\leq \sum_{k>r-|C|/2}||h_k-\mathcal{P}_{R+q} (h_k)||\leq 2K\sum_{j\geq r+q} e^{-j/(2\xi)}\leq O(e^{-(r+q)/(2\xi)}).
	\end{equation}
	$S=L$ can be treated entirely analogously.
\end{proof}
We have completed the split into local terms and tails. The intuition is now that the entanglement in the system will be driven on the one hand by the term $H_{\del}$, which is the only local term connecting $A$ and $B$, and on the other hand by the tail. The former can only create entanglement locally around $C$, and the latter is weak. Next, we need to formalize the notion that the tails can only create weak entanglement.

\begin{lemma}[Tails create weak entanglement.]\label{lem:IEB}
	Let $R$ be any region of the chain, then
	\begin{equation}
		||(U_S(t)\rho (U_S^T(t))^{\dagger})^{T_R}||_1\leq O\left(\exp\left(2e^{-r/(2\xi)} C_{R,S}t\right)\right)||\rho^{T_R}||_1
	\end{equation}
	with $C_{R,S}=\sum_{q=0}^{\infty} e^{-q/(2\xi)}d_{R\cap (S+q)}^2$.
\end{lemma}
We note in passing that clearly this bound is non trivial only if the constant $C_{R,S}$ is finite.

\begin{proof}
	First, notice the following fact. For any region $R$, operator $O$, and local operator $Z_I$ on a region 
	$I$, we have
	\begin{equation}\label{eq:ent_neg_product}
		\|(Z_I O)^{T_R}\|_1\leq d_{R\cap I}\|Z_I O^{T_{R/I}}\|_1\leq d_{R\cap I}\|Z_I\|\| O^{T_{R/I}}\|_1= d_{R\cap I}\|Z_I\|\| (O^{T_{R}})^{T_{R\cap I}}\|_1\leq d_{R\cap I}^2\|Z_I\|\| O^{T_{R}}\|_1
	\end{equation}
	next, we use the Dyson series expression for the tail evolution. First, notice that by the triangle inequality
	\begin{equation}
		\begin{aligned}
			&||(U_S(t)\rho (U_S^T(t))^{\dagger})^{T_R}||_1\leq \sum_{k,n=0}^{\infty} \int_0^t\mathrm dt_1\dots\int_0^{t_{k-1}}\mathrm dt_k\int_0^t\mathrm ds_1\dots\int_0^{s_{n-1}}\mathrm ds_n\left\|\left(\tilde T_S(t_1)\dots \tilde T_S(t_k)\rho\tilde T_S(s_n)\dots \tilde T_S(s_1)\right)^{T_R}\right\|_1
		\end{aligned}
	\end{equation}
	recall that we can write \begin{equation}
		\tilde T_S(t)=\sum_{q=0}^{\infty} \tilde Z_{S+q}(t)
	\end{equation}
	where $\tilde Z_{S+q}(t)=e^{-iH_St}Z_{S+q}e^{iH_St}$ acts only on $S+q$. We then have, using Eq.~ \eqref{eq:ent_neg_product} and the triangle inequality
	\begin{equation}
		\begin{aligned}
			&\left\|\left(\tilde T_S(t_1)\dots \tilde T_S(t_k)\rho\tilde T_S(s_n)\dots \tilde T_S(s_1)\right)^{T_R}\right\|_1\leq\\& \sum_{\substack{q_1,\dots, q_k=0\\p_1,\dots, p_n=0}}^{\infty} \left\|\left(\tilde Z_{S+q_1}(t_1)\dots \tilde Z_{S+q_k}(t_k)\rho\tilde Z_{S+p_n}(s_n)\dots \tilde Z_{S+p_1}(s_1)\right)^{T_R}\right\|_1 \leq ||\rho^{T_R}||_1\left(\sum_{q=0}^{\infty} ||Z_{S+q}||d_{R\cap(S+q)}^2\right)^{n+k} \\&
		\end{aligned}
	\end{equation}
	before plugging this in in the initial expression, notice that 
	\begin{equation}
		\int_0^t\mathrm dt_1\dots\int_0^{t_{k-1}}\mathrm dt_k= \frac{t^k}{k!},
	\end{equation}
	we then have
	\begin{equation}
		\begin{aligned}
			&||(U_S^T(t)\rho (U_S^T(t))^{\dagger})^{T_R}||_1\leq \sum_{k,n=0}^{\infty} \frac{t^k}{k!}\frac{t^n}{n!} \left(\sum_{q=0}^{\infty} ||Z_{S+q}||d_{R\cap(S+q)}^2\right)^{n+k}=\exp\left(2t \sum_{q=0}^{\infty} ||Z_{S+q}||d_{R\cap(S+q)}^2\right).
		\end{aligned}
	\end{equation}
	We can
	conclude by using that by Lemma \ref{lem:tail} 
	that this expression has a norm bounded by $\|\tilde Z_{S+q}(t)\|=\| Z_{S+q}\|\leq O\left(e^{-(r+q)/(2\xi)}\right)$.
\end{proof}
This allows us to prove the first part of the bound: We can write
\begin{equation}
	e^{iHt}=  e^{iH_{\del} t}U^{T}_{\del}(t)e^{iH_Lt}U^{T}_L(t)e^{iH_Rt}U^{T}_R(t) .
\end{equation}
Given a generator $h$, define the evolution operator $\mathcal U_h(\rho):=e^{ih}\rho e^{-ih}$, furthermore define the tail evolution operator as $\mathcal U^{T,t}_S(\rho)=U^{T}(t)_S\rho (U^T(t)_S)^{\dagger}$. Consider a state $\sigma$, notice that the overlap between $B$ and $\del$ is of length $2r-|C|$, in particular, 
\begin{equation}
	||\tr_C(\mathcal U_{H_{\del}t}(\sigma))^{T_B}||_1\leq ||\tr_C(e^{iH_{\del}t}(\sigma)^{T_{(2r-|C|/2,\infty)}} e^{-iH_{\del}t})^{T_{[|C|/2, 2r-|C|/2]}}||_1 \leq 2^{2r-|C|} ||\sigma^{T_{(2r-|C|/2,\infty)}}||_1
\end{equation}
Here, we simply assume sufficient time has passed that $H_{\del}$ has had time to saturate the entanglement in the region $\del$, that is, create a maximally entangled state between the sites on either side of $C$. We have then
\begin{equation}
	\begin{aligned}
		||\rho(t)^{T_B}||_1&=\left\|\tr_C\left(\mathcal U_{H_{\del}t}\circ \mathcal U_{\del}^{T,t}\circ \mathcal U_{{H_L}t}\circ \mathcal U_{L}^{T,t}\circ \mathcal U_{H_{R} t}\circ \mathcal U_{R}^{T,t}(\rho)\right)^{T_B}\right\|_1\\&\leq 2^{2r-|C|}\left\|\left(\mathcal U_{\del}^{T,t}\circ \mathcal U_{{H_L}t}\circ \mathcal U_{L}^{T,t}\circ \mathcal U_{H_{R} t}\circ \mathcal U_{R}^{T,t}(\rho)\right)^{T_{(2r-|C|/2,\infty)}}\right\|_1 .
	\end{aligned}
\end{equation}
Next, by Lemma \ref{lem:IEB}, 
we have 
\begin{equation}
	\begin{aligned}         
		\left\|\rho(t)^{T_B}\right\|_1&\leq 2^{2r-|C|}E_{\del, t}\left\| \left(\mathcal U_{{H_L}t}\circ \mathcal U_{L}^{T,t}\circ \mathcal U_{H_{R} t}\circ \mathcal U_{R}^{T,t}(\rho)\right)^{T_{(2r-|C|/2,\infty)}}\right\|_1
	\end{aligned}
\end{equation}
where 
\begin{equation}
	\begin{aligned}
		E_{\del,t}&\leq\exp\left( t O\left( e^{-r/(2\xi)}\sum_{q=0}^{\infty}2^{2(|\del+q+1 \cap(2r-|C|/2,\infty)| )}e^{-q/(2\xi)}\right)\right)\\&\leq \exp\left( t O\left(e^{-r/(2\xi)}\sum_{q=0}^{\infty}2^{2q }e^{-q/(2\xi)}\right)\right)\leq \exp\left(tO\left(e^{-r/(2\xi)}\right)\right),
	\end{aligned}
\end{equation}
where we have used the assumption $\xi< {1}/({4\log(2)})$, so that the terms in the sum above are exponentially decaying and sum to a constant. Since $L$ has no overlap with $(2r-|C|/2,\infty)$, we can just bring the partial transpose inside the $L$ evolution and eliminate it by unitarity. Afterwards, we 
eliminate $\mathcal U_{L}^{T,t}$, exactly as we 
have eliminated $\mathcal U_{\del}^{T,t}$, to get
\begin{equation}
	\begin{aligned}
		||\rho(t)^{T_B}||_1&= 2^{2r-|C|}E_{\del, t}E_{L, t}\left\| \left( \mathcal U_{H_{R} t}\circ \mathcal U_{R}^{T,t}(\rho)\right)^{T_{(2r-|C|/2,\infty)}}\right\|_1
	\end{aligned}
\end{equation}
and in the same way one can show that $E_{L, t}\leq \exp\left(t\, O\left(e^{-r/(2\xi)}\right)\right)$. Finally, notice that for any region $X$ and a state $\sigma$, $||\sigma^{T_X}||_1=||\sigma^{T_{X^c}}||_1$, since $\sigma$ is Hermitian. Then 
\begin{equation}
	\begin{aligned}
		\left\| \left( \mathcal U_{H_{R} t}\circ \mathcal U_{R}^{T,t}(\rho)\right)^{T_{(2r-|C|/2,\infty)}}\right\|_1&= \left\|\left(\mathcal U_{H_{R} t}\circ \mathcal U_{R}^{T,t}(\rho)\right)^{T_{(-\infty,2r-|C|/2)}}\right\|_1\\&\leq 2^{2r-|C|}\left\|\left(\mathcal U_{R}^{T,t}(\rho)\right)^{T_{(-\infty,-|C|/2)}}\right\|_1\leq 2^{2r-|C|}E_{R,t}\left\|\rho^{T_{(-\infty,-|C|/2)}}\right\|_1\\&=2^{2r-|C|}E_{R,t},
	\end{aligned}
\end{equation}
where we have used that the initial state is a product state, $E_{R,t}$ 
can be bounded as in the previous cases. Overall, we have then proven
that
\begin{equation}
	||\rho(t)^{T_B}||_1\leq 2^{4r-2|C|}\exp\left(t\,O\left(e^{-r/(2\xi)}\right)\right)
\end{equation}
holds true. 
The final result is obtained by applying $\log_2$ on both sides. Recall that the final bound for the growth in time of the logarithmic negativity is obtained by substituting $r\sim \log(t)$ when appropriate, see Corollary \ref{cor:final_bound}.
\bibliography{refs}
\end{document}